\documentclass[10pt,a4paper,openright]{scrartcl}
\KOMAoption{captions}{tableheading}
\usepackage{subfig}
\usepackage[english]{babel}
\usepackage{amsmath, amsthm, amssymb}
\usepackage{bbm}
\usepackage{url}
\usepackage{tikz-cd}
\usepackage{braket}
\usepackage{mathrsfs}
\usepackage{path}
\usepackage{nicefrac}
\usepackage{cite}
\usepackage{cancel}
\usepackage{indentfirst}
\usepackage{appendix}
\usepackage[utf8]{inputenc}
\usepackage{geometry}
\usepackage{verbatim}
\usepackage{nameref}
\usepackage{alltt}
\usepackage[colorlinks=true, linkcolor=violet, citecolor=violet, urlcolor=violet]{hyperref}
\usepackage{pgfplots}
\usepackage{dsfont}
\usepackage[font=small,labelfont=bf]{caption} 

\newtheorem{prop}{Proposition}
\newtheorem{Theo}{Theorem}

\usepackage{etoolbox}
\theoremstyle{definition}
\newtheorem{rem}{Remark}
\AtEndEnvironment{rem}{\hfill$\lozenge$}

\newcommand\be{\begin{equation}}
\newcommand\ee{\end{equation}}

\usepackage{color}

\areaset[current]{500pt}{680pt}
\addtokomafont{disposition}{\normalfont\bfseries}

\numberwithin{equation}{section}
\begin{document}

\ 
\bigskip

\thispagestyle{empty}
\begin{center}
	\Large{\textbf{Haantjes algebras, Zernike system and separation of variables}}
\end{center}
\vskip 0.5cm
\begin{center}
	\textsc{Ond\v{r}ej Kub\r{u}$^{1,*}$ and Danilo Latini$^{2,3,\star}$}
\end{center}
\begin{center}
	$^1$ Instituto de Ciencias Matem\'aticas, C/ Nicol\'as Cabrera, No 13–15, 28049 Madrid,
	Spain
\end{center}
\begin{center}	
	$^2$ Universit\`a degli Studi di Milano, Dipartimento di Matematica ``Federigo Enriques", \\
	Via Cesare Saldini 50, 20133 Milano, Italy
\end{center}
\begin{center}
	$^3$ INFN Sezione di Milano, Via Giovanni Celoria 16, 20133 Milano, Italy 
\end{center}
\begin{center}
	\footnotesize{
		$^*$\textsf{ondrej.kubu@icmat.es}  \hskip 0.25cm $^\star$\textsf{danilo.latini@unimi.it} }
\end{center}
\vskip 0.75cm
\begin{center}
	\textbf{Abstract}
\end{center}
\begin{abstract}

\noindent The generalized Zernike family $H_{(N)} = p_1^2 + p_2^2 +
\sum_{n=1}^N \gamma_n\,(q_1 p_1 + q_2 p_2)^n$ is a parametric family
of two-dimensional superintegrable Hamiltonians, admitting $N$ integrals
of motion of degree $N$ in the momenta. A theorem of Nozaleda, Tempesta,
and Tondo guarantees that canonical separation coordinates
(Darboux--Haantjes coordinates) exist for any such system; the challenge
is to construct them explicitly. This paper solves the problem for
$N = 2$ --- the classical Zernike system, which is canonically equivalent
to the isotropic harmonic oscillator on flat space or on a space of
constant curvature --- covering all four known separation types: polar,
two Cartesian-type, and elliptic.

The key structural fact is that the Haantjes operators associated with
all integrals of $H_{(2)}$ have no momentum-dependent off-diagonal block
(lift form). We prove that this implies the separation coordinates are
reachable by an extended point transformation: the new positions depend
only on the old positions, with no momentum entering the coordinate
change. In the polar and Cartesian-type cases the new position
coordinates involve at most a square root of a single-variable rational
function; in the elliptic case they are given by the two roots of a
quadratic polynomial in the original coordinates, and the resulting
branch structure introduces a fourth regular singular point in the
quantum separated ODE, placing it in the Heun class, in agreement with
results of Atakishiyev, Pogosyan, Vicent, Wolf, and Yakhno.

For $N \geq 3$ we prove an obstruction: no lift-form Haantjes operator
can generate an integral independent of the angular momentum. The
separation coordinates for higher Zernike Hamiltonians therefore require momentum-dependent canonical transformations, whose explicit
construction is the subject of future work.
\end{abstract}
\vskip 0.35cm
\hrule

\medskip

\medskip

\bigskip
\hrule

\tableofcontents


\section{Introduction}
\label{intro}

\noindent Superintegrable Hamiltonian systems, i.e. systems with more integrals of motion than degrees of freedom, are well understood in the case of second-order polynomial integrals: in two and three dimensions this class is classified \cite{MPW, KKMbook}. Far less is known about systems with higher-order integrals, and the situation is especially rich — and open — for families of Hamiltonians where the degree of the integrals grows with a parameter. The Tremblay–Turbiner–Winternitz (TTW) \cite{TTW09} and Post–Winternitz (PW) \cite{PW10} families, where the degree depends on a rational parameter $k$, are the most studied examples.

The focus of our article is the \emph{generalized Zernike family}
\begin{equation}
\label{eq:HN}
H_{(N)} := p_1^2 + p_2^2 + \sum_{n=1}^N \gamma_n\,(q_1 p_1 + q_2 p_2)^n, \qquad N \geq 1,
\end{equation}
where each coefficient $\gamma_i$ is real or purely imaginary (the two types may differ between different $i$). This is another interesting example: it is superintegrable for every $N$ and every choice of parameters $\gamma_n$, and its integrals of motion are polynomial of degree growing with $N$ \cite{AGSH}. The case $N = 1$ is the (complex) isotropic harmonic oscillator, $N = 2$ is the \emph{classical Zernike Hamiltonian} introduced by Pogosyan, Wolf, and Yakhno \cite{PWY2017} as the classical counterpart of a differential operator studied by Zernike \cite{Zernike} in the context of optical diffraction theory, and $N \geq 3$ gives a hierarchy of higher-order superintegrable systems.

The superintegrability of $H_{(N)}$ was established in the polynomial case by Blasco, Gutierrez-Sagredo, and Herranz \cite{AGSH}, and for arbitrary analytic $F(\mathbf{q}\cdot\mathbf{p})$ by C. Gonera, J. Gonera, and Kosiński \cite{Gonera22}, who also showed that the $n$-dimensional generalisation $H = \mathbf{p}^2 + F(\mathbf{q}\cdot\mathbf{p})$ is maximally superintegrable in $\mathbb{R}^n$. For $N = 2$, the system admits the angular momentum $\mathcal{J} = q_1 p_2 - q_2 p_1$ and two quadratic integrals $I_1$, $I_2$ (see Proposition~\ref{prop:myprop}), which together close in a \emph{cubic Higgs-type algebra} \cite{PWY2017, PSAWY2017}; for general $N$ the symmetry algebra is a polynomial algebra of degree $2N - 1$ \cite{AGSH, bulg}. A further structural feature of $H_{(2)}$ is its equivalence with the isotropic harmonic oscillator on spaces of constant curvature: through a canonical transformation to geodesic polar coordinates \cite{AGSH, ANY26}, $H_{(2)}$ becomes the Higgs oscillator \cite{Higgs79} on $\mathbb{S}^2$ ($\kappa := -\gamma_2 > 0$), the flat oscillator on $\mathbb{E}^2$ ($\kappa = 0$), or its hyperbolic analogue on $\mathbb{H}^2$ ($\kappa < 0$); see \cite{Fordy2018} for an overview.

Hamilton--Jacobi separation of variables for $H_{(2)}$ has been studied in several coordinate systems. Pogosyan, Wolf, and Yakhno \cite{PWY2017} identified separations in polar, equidistant, and elliptic coordinates via vertical projection from $\mathbb{S}^2$ and $\mathbb{H}^2$. The quantum system admits analogous separations \cite{PSAWY2017, Campoamor25}, with solutions in terms of Legendre, Gegenbauer, and Jacobi polynomials \cite{Pog17, Ataki17}; interbasis expansion coefficients and connections to hyperspherical harmonics on $\mathbb{S}^3$ were computed in \cite{Ataki17, Ataki19}, and elliptic separations on the half-sphere --- yielding Heun-type equations --- were studied in \cite{Ataki18, Ataki19b}. All these results relied on the geometric identification of $H_{(2)}$ with the Higgs oscillator on spaces of constant curvature — already implicit in Zernike's original optical context — which guided the construction of separation coordinates via sphere and hyperboloid projections. For $N \geq 3$, no analogous curved-space model is known, and this geometric route is unavailable \cite{AGSH, Gonera22}.

The standard framework for separation of variables, the Stäckel--Benenti theory \cite{BenentiFrancaviglia1980, Benenti2016}, operates on configuration space and is intrinsically adapted to second-order integrals: it characterises separability through Killing tensors on a Riemannian manifold. For the family $H_{(N)}$ with $N \geq 3$, whose integrals grow in polynomial degree with $N$, this theory provides no systematic handle.

The theory of \emph{Haantjes algebras} \cite{haantjes1955, tempesta2021haantjes, tempesta2022haantjes, nozaleda2022classical} offers a different starting point, generalizing the $\omega N$-manifold theory of Magri and Morosi \cite{magri1984} and the quasi-bi-Hamiltonian approach of Morosi and Tondo \cite{morosi1997, tondo1999}: it acts on the full phase space and handles integrals of arbitrary degree, or (in principle) even non-polynomial ones. Building on Haantjes' classical result \cite{haantjes1955}, Reyes Nozaleda, Tempesta, and Tondo \cite{tempesta2021haantjes,tempesta2022haantjes,nozaleda2022classical} established a separability criterion in terms of a semisimple Abelian $\omega\mathscr{H}$ structure and gave a constructive algorithm — the Darboux–Haantjes (DH) algorithm — that produces separation coordinates based on eigendistributions of the Haantjes operators; Haantjes chains are the Haantjes-algebraic generalisation of Lenard chains \cite{magri2003}. The Drach--Holt system, which is integrable due to a cubic integral \cite{Campoamor13}, serves as a proof of concept for the higher-order case \cite{nozaleda2022classical}. In this paper we apply this framework to $H_{(2)}$, as a first step toward the full family $H_{(N)}$, $N \geq 3$.

The central question addressed in this paper is: \emph{for which $N$ can the DH separation coordinates be reached by an extended point transformation (EPT)}, i.e.\ a canonical transformation in which the new positions depend only on the old positions? The answer partitions the generalized Zernike family into two regimes:

\begin{itemize}
\item \textbf{$N \leq 2$ (EPT regime).}
For $N=1$ (harmonic oscillator) separation in Cartesian coordinates is trivially an EPT.  For $N=2$, we show in Section~\ref{sec4} that all four separations of $H_{(2)}$ can be reached by EPTs: the polar separation associated with $\mathcal{J}$ (which holds for the entire family $\{H_{(N)}\}_{N\geq 1}$ simultaneously), the two Cartesian-type separations associated with $I_1$ and $I_2$, and the elliptic separation associated with $I_e = -\gamma_2 k_1^2 \mathcal{J}^2 - k_2^2 I_2$.  In each case, Proposition~\ref{prop:a_block} shows that the DH coordinates are determined by the position-space block of the Haantjes operator alone.  The Cartesian and elliptic coordinates coincide with the System~II/HII and elliptic coordinates of Pogosyan, Wolf, and Yakhno \cite{PWY2017}, here derived by purely algebraic means.

\item \textbf{$N \geq 3$ (beyond EPT).}
Theorem~\ref{thm:no_ept} in Section~\ref{sec:no_ept} establishes that for $N \geq 3$ with $\gamma_k \neq 0$ for some $k\geq 3$, any EPT that separates $H_{(N)}$ must be a transformation to polar coordinates associated to $\mathcal{J}^2$; in particular, the Haantjes operators corresponding to $I_1$ and $I_2$ cannot be diagonalised by any EPT.  The Jacobi--Haantjes theorem from \cite{nozaleda2022classical} nevertheless guarantees the existence of DH separation coordinates via some canonical transformation; their explicit construction uses a priori momentum-dependent transformations mixing coordinates and momenta, and is the main open problem motivating the present work.
\end{itemize}

\section{The classical superintegrable Zernike system}
\label{sec2}

\subsection{The model, the integrals of motion and the symmetry algebra}
\label{sec2.1}

\noindent In this work, we are interested in the Hamiltonian
\begin{equation}
H_{(2)}:=p_1^2+p_2^2+\gamma_1 (q_1 p_1 + q_2 p_2)+\gamma_2 (q_1 p_1 + q_2 p_2)^2
\label{eq:hamZK}
\end{equation}
where $(q_i,p_i)$ denote pairs of canonical coordinates and the subscript $(2)$ in $H_{(2)}$ indicates the highest power of $q_1 p_1 + q_2 p_2$ appearing in the Hamiltonian. More precisely, for $\gamma_1=-\imath\,\beta$ and $\gamma_2=\alpha$, the system reduces to the \emph{proper} Zernike Hamiltonian $H_{\textsf{Zk}}$, introduced in \cite{PWY2017} as the classical counterpart of the Zernike operator \cite{Zernike}, and studied in detail in \cite{PSAWY2017} in the quantum setting. The superintegrability property of the classical Hamiltonian above is characterized through the following:
\begin{prop}[\cite{PWY2017}]\label{prop:myprop}
For any values of the parameter $\gamma_1$, $\gamma_2$, the Hamiltonian \eqref{eq:hamZK} Poisson commutes with the angular momentum $\mathcal{J}=q_1 p_2 -q_2 p_1$ and with the two additional functionally independent constants of motion
\begin{equation}
I_1=\left(1+\gamma_2 (q_1^2+q_2^2)\right)p_1^2+\gamma_1 q_1 p_1 \, , \qquad I_2=\left(1+\gamma_2 (q_1^2+q_2^2)\right)p_2^2+\gamma_1 q_2 p_2 \, .
\label{eq:addconst}
\end{equation}
Since we are in dimension two, there is, of course, a dependence relation
\begin{equation}
H_{(2)}=I_1+I_2-\gamma_2 \mathcal{J}^2 \, ,
\label{eq:dep}
\end{equation}
which finally establishes the superintegrability of the system: in fact, the two sets $\{H_{(2)}, \mathcal{J}, I_1\}$ and $\{H_{(2)}, \mathcal{J}, I_2\}$ are formed by three functionally independent constants of motion. Concerning the symmetry algebra, the following first integrals
\begin{equation}
X_1:=\mathcal{J}/2 \, , \qquad X_2:=(I_1-I_2)/2 \, , \qquad X_3:=\{X_1, X_2\}
\label{integralsofmotion}
\end{equation}
close in a cubic Higgs-type algebra \cite{PWY2017}
\begin{equation}
\{X_1, X_2\}= X_3 \, , \qquad \{X_3, X_1\}= X_2 \, , \qquad \{X_2, X_3\}= -(\gamma_1^2+2 \gamma_2 H_{(2)})X_1-8 \gamma_2^2 X_1^3
\label{cubhiggs}
\end{equation}
which is endowed with a quartic Casimir invariant
\begin{equation}
C=X_2^2+X_3^2-(\gamma_1^2+2 \gamma_2 H_{(2)})X_1^2-4 \gamma_2^2 X_1^4 \, ,
\label{casimir}
\end{equation}
which, in the realization, reads
\begin{equation}
C=\frac{1}{4}H_{(2)}^2 \, .
\label{eq:cas}
\end{equation}
\end{prop}
\begin{rem}
It is interesting to note that, in terms of the Casimir \eqref{casimir}, the symmetry algebra can be rewritten as
\begin{equation}
	\{X_1, X_2\}= \frac{1}{2}\frac{\partial C}{\partial X_3} \, , \qquad \{X_3, X_1\}=  \frac{1}{2}\frac{\partial C}{\partial  X_2} \, , \qquad \{X_2, X_3\}=  \frac{1}{2}\frac{\partial C}{\partial  X_1} \, .
	\label{cubhiggss}
\end{equation}
This expresses each bracket as half the partial derivative of the Casimir $C$ with respect to the missing generator, a pattern analogous to the representation of quadratic Poisson algebras via the polynomial $h$ in \cite[eqs.~(8)--(9)]{Das01}.
\end{rem}
\noindent An important special case arises when $\gamma_2=\alpha=0$ and $\gamma_1=-\imath \beta$, as in the original paper mentioned before, giving the system
 \begin{equation}
 H_{(1)}=p_1^2+p_2^2-\imath \beta (q_1 p_1+q_2 p_2) \, ,
 \label{eq:ho2}
 \end{equation}
 which is canonically equivalent to the $2$D isotropic harmonic oscillator, via the canonical transformation
 \begin{equation}
	\label{Eq:canonic}
	\begin{split}
 q_1&=\sqrt{2}\bar{q}_1 \, , \qquad p_1=\frac{1}{\sqrt{2}} \bar{p}_1+\imath \frac{\beta}{\sqrt{2}}\bar{q}_1 \, ,\\
 q_2&=\sqrt{2}\bar{q}_2 \, , \qquad p_2=\frac{1}{\sqrt{2}} \bar{p}_2+\imath \frac{\beta}{\sqrt{2}}\bar{q}_2 \, ,
	\end{split}
 \end{equation}

\noindent which maps the system to
 \begin{equation}
	\bar{H}_{(1)}=\frac{\bar{p}_1^2+\bar{p}_2^2}{2}+\frac{\beta^2}{2}(\bar{q}_1^2+\bar{q}_2^2) \, ,
	\label{eq:ho}
\end{equation}
and the integrals of motion to
\begin{equation}
	\bar{X}_1=\bar{\mathcal{J}}/2 \, , \qquad \bar{X}_2:=(\bar{I}_1-\bar{I}_2)/2 \, , \qquad \bar{X}_3:=\{\bar{X}_1, \bar{X}_2\}
	\label{integrals}
\end{equation}
with
\begin{equation}
\bar{X}_1=\frac{1}{2}(\bar{q}_1 \bar{p}_2-\bar{q}_2 \bar{p}_1) \, , \qquad \bar{X}_2=\frac{1}{4}(\bar{p}_1^2+ \beta^2\bar{q}_1^2-\bar{p}_2^2-\beta^2 \bar{q}_2^2) \, , \qquad \bar{X}_3=\frac{1}{2}(\bar{p}_1 \bar{p}_2+\beta^2 \bar{q}_1 \bar{q}_2)
\end{equation}
which close in the standard algebra of the harmonic oscillator on $\mathbb{E}^2$
\begin{equation}
	\{\bar{X}_1, \bar{X}_2\}= \bar{X}_3 \, , \qquad \{\bar{X}_3, \bar{X}_1\}= \bar{X}_2 \, , \qquad \{\bar{X}_2, \bar{X}_3\}= \beta^2 \bar{X}_1
	\label{oscymm}
\end{equation}
together with the corresponding Casimir function
\begin{equation}
C=\beta^2 \bar{X}_1^2+\bar{X}_2^2+\bar{X}_3^2 \, .
\label{eq:casimosc}
\end{equation}

\subsection{On the connection with the oscillator on \texorpdfstring{$\mathbb{S}^2$ and $\mathbb{H}^2$}{S2 and H2}}
\label{sec2.2}
\noindent There exists a natural relationship between the previous interpretation of $H_{(2)}$ on $\mathbb{E}^2$ and an alternative one as a superintegrable Hamiltonian on a $2$D positive or negative constant curvature space \cite{Fordy2018, AGSH,ANY26}. To this aim, let us consider polar coordinates
\begin{align}
q_1 &= r \cos (\varphi)\, , \qquad p_1 = \cos(\varphi) p_r-\frac{\sin(\varphi)}{r} p_\varphi\\
q_2 &= r \sin (\varphi)\, , \qquad p_2 = \sin(\varphi) p_r+\frac{\cos(\varphi)}{r} p_\varphi \, ,
\label{polcoor}
\end{align}
for $r>0$, $\varphi \in [0,2\pi)$. In these variables, the original Hamiltonian \eqref{eq:hamZK} assumes the following form
\begin{equation}
H=p_r^2+\frac{p_\varphi^2}{r^2}+\gamma_1 (r p_r) +\gamma_2 (r p_r)^2=(1+\gamma_2 r^2)p_r^2+\frac{p_\varphi^2}{r^2}+\gamma_1 (r p_r) 
\label{eq:ham2.1}
\end{equation}
In these variables, the $2$D metric reads
\begin{equation}
\text{d}s^2=\frac{1}{1+\gamma_2 r^2}\text{d}r^2+r^2 \text{d}\varphi^2 \, .
\label{eq:metric}
\end{equation}
The Gaussian curvature is constant, $\kappa = -\gamma_2$ \cite{AGSH,Fordy2018,bulg}, and the metric~\eqref{eq:metric} describes the three classical simply connected two-dimensional spaces of constant curvature:
\begin{align}
	\kappa &= -\gamma_2 = 0 
	\quad \Longleftrightarrow \quad \mathbb{E}^2 \, ,\\
		\kappa &= -\gamma_2 > 0 
	\quad \Longleftrightarrow \quad \mathbb{S}^2 \, ,\\
	\kappa &= -\gamma_2 \,< 0 
	\quad \!\Longleftrightarrow \quad \mathbb{H}^2 \, .
\end{align}
Now, let us notice that the polar radial coordinate $r$ is no longer a geodesic distance in a
curved space with $\kappa \neq 0$ and, in order to perform an appropriate geometrical interpretation of the system on the sphere and hyperbolic space,
we need to introduce the so-called geodesic radial
coordinates (see \cite{AGSH} and references therein). Taking into account $\kappa=-\gamma_2$, we introduce
\begin{equation}
r=S_{\kappa}(\rho) \, ,
\end{equation}
where $\rho$ denotes the geodesic distance from the origin to the particle, while $\varphi$ is the usual angular coordinate. Here, we have introduced the $\kappa$-dependent sine functions, which together with the cosine functions are defined as
\begin{equation}
	S_{\kappa}(x):=\begin{cases}
		\dfrac{1}{\sqrt{\kappa}}\sin\!\bigl(\sqrt{\kappa}\,x\bigr), & \kappa>0\, , \\[4pt]
		x, & \kappa=0\, , \\[4pt]
		\dfrac{1}{\sqrt{-\kappa}}\sinh\!\bigl(\sqrt{-\kappa}\,x\bigr), & \kappa<0\, .
	\end{cases} \qquad \qquad 	C_{\kappa}(x):=
	\begin{cases}
	\cos\!\bigl(\sqrt{\kappa}\,x\bigr), & \kappa>0\, , \\[4pt]
	1, & \kappa=0\, , \\[4pt]
	\cosh\!\bigl(\sqrt{-\kappa}\,x\bigr), & \kappa<0\, .
	\end{cases}
\label{curvsine}
\end{equation}
The $\kappa$-dependent tangent function is defined as
\begin{equation}
T_{\kappa}(x):=\frac{S_{\kappa}(x)}{C_{\kappa}(x)}\,.
\label{ktan}
\end{equation}
The metric in these new coordinates reads
\begin{equation}
\text{d}s^2=\text{d}\rho^2+S_{\kappa}^2(\rho) \text{d}\varphi^2
\label{eq:metrix}
\end{equation}
Let us notice that when $\kappa=-\gamma_2=0$, i.e. in the flat limit, the metric \eqref{eq:metrix} collapses to
\begin{equation}
\text{d}s^2=\text{d}r^2+r^2 \text{d}\varphi^2
\label{flatmet}
\end{equation}
since in this regime $\rho=r$. Introducing canonical geodesic polar variables $\{\rho, \varphi, p_\rho, p_\varphi\}$, the Zernike Hamiltonian \eqref{eq:hamZK} takes the form of a natural Hamiltonian under the following canonical transformation \cite{AGSH}
\begin{align}
	q_1 &= S_{\kappa}(\rho)\cos\varphi \, , \qquad 
	p_1 =
	\frac{\cos\varphi}{C_{\kappa}(\rho)}
	\left(
	p_{\rho} - \frac{\gamma_1}{2} T_{\kappa}(\rho)
	\right)
	-	\frac{\sin\varphi}{S_{\kappa}(\rho)}\, p_{\varphi} \, , \\
		q_2 &= S_{\kappa}(\rho)\sin\varphi \, , \qquad p_2 =
	\frac{\sin\varphi}{C_{\kappa}(\rho)}
	\left(
	p_{\rho} - \frac{\gamma_1}{2} T_{\kappa}(\rho)
	\right)
	+
	\frac{\cos\varphi}{S_{\kappa}(\rho)}\, p_{\varphi} \, ,
\end{align}
with $0 < \rho < \frac{\pi}{2\sqrt{\kappa}}$ when $\kappa>0$, i.e. when we are considering the $2$-sphere $\mathbb{S}^2$, or $0<\rho<\infty$ when $\kappa <0$, i.e. when we are considering the hyperboloid $\mathbb{H}^2$. In these variables the Hamiltonian assumes the form
\begin{equation}
H_{(2)}=p_\rho^2+\frac{p_{\varphi}^2}{S^2_\kappa(\rho)}-\frac{\gamma_1^2}{4}T_\kappa^2(\rho) \, .
\label{eq:hamigeorada}
\end{equation}
This Hamiltonian, depending on the value of $\gamma_1$ and $\kappa=-\gamma_2$, describes different systems. When $\gamma_1 = 2\imath\omega$ with $\omega \in \mathbb{R}$, one has
\begin{equation}
	H_{(2)}=p_\rho^2+\frac{p_{\varphi}^2}{S^2_\kappa(\rho)}+\omega^2T_\kappa^2(\rho) \, =p_{\rho}^{2}
	+\frac{\kappa\,p_{\varphi}^{2}}{\sin^{2}\!\bigl(\sqrt{\kappa}\,\rho\bigr)}
	+\frac{\omega^{2}}{\kappa}\tan^{2}\!\bigl(\sqrt{\kappa}\,\rho\bigr)
	\label{eq:hamigeorad}
\end{equation}
In this explicit form, the Hamiltonian covers the Higgs oscillator \cite{Higgs79, ANY26} on $\mathbb{S}^2$
($\kappa > 0$), the hyperbolic one on $\mathbb{H}^2$
($\kappa < 0$) and the usual isotropic harmonic oscillator
on $\mathbb{E}^2$ ($\kappa = -\gamma_2 = 0$) \cite{bulg}.

\begin{rem}
The canonical transformation above decomposes into two steps, both discussed in \cite{AGSH}.
The first is the passage to geodesic polar coordinates $(\rho,\varphi)$, which accounts for the
curvature $\kappa=-\gamma_2$ of the underlying space.
The second is a gauge shift $p_\rho \mapsto p_\rho - \tfrac{\gamma_1}{2}T_\kappa(\rho)$ in the
fiber, which removes the term linear in momenta and converts the Hamiltonian into natural
(kinetic-plus-potential) form \eqref{eq:hamigeorada}.
This shift is globally well-defined because the magnetic $1$-form $\tfrac{\gamma_1}{2}\,\mathrm{d}(\ln r)$
is exact, i.e.\ its curvature vanishes.
The precise sense in which both steps are invisible to the Haantjes operators of Section~\ref{sec4}
is spelled out in Remark~\ref{rem:curvature_blind}.
\end{rem}

\section{Separation of variables via Haantjes geometry}
\label{sec3}
\subsection{Haantjes, Nijenhuis tensors, \texorpdfstring{$\omega \mathscr{H}$}{omegaH} manifolds and Haantjes chains}
\label{subsec3.1}

The theory of Haantjes operators provides a geometric framework for identifying separation coordinates in Hamiltonian systems. Following the approach developed in \cite{tempesta2021haantjes,tempesta2022haantjes,nozaleda2022classical}, we outline the key theoretical constructs and the computational procedure.

\medskip

\noindent \textbf{Nijenhuis and Haantjes tensors.}
Let $M$ be a smooth $n$-dimensional manifold and $L : TM \to TM$ a $(1,1)$ tensor field. The \emph{Nijenhuis torsion} of $L$ is defined by
\be
\label{eq:nijenhuis_torsion}
T_L(X,Y) := L^2[X,Y] + [LX, LY] - L([X, LY] + [LX,Y]),
\ee
where $X,Y \in TM$ and $[\cdot,\cdot]$ denotes the Lie bracket. In local coordinates $x = (x^1,\ldots,x^n)$, the Nijenhuis torsion is the skew-symmetric $(1,2)$ tensor
\be
(T_L)^i_{jk} = \sum_{\alpha=1}^n \left( \frac{\partial L^i_k}{\partial x^\alpha} L^\alpha_j - \frac{\partial L^i_j}{\partial x^\alpha} L^\alpha_k + \left(\frac{\partial L^\alpha_j}{\partial x^k} - \frac{\partial L^\alpha_k}{\partial x^j}\right) L^i_\alpha \right).
\ee

The \emph{Haantjes torsion} of $L$ is the vector-valued $2$-form
\be
\mathcal{H}_L(X,Y) := L^2 T_L(X,Y) + T_L(LX, LY) - L(T_L(X, LY) + T_L(LX,Y)),
\ee
with local expression
\be
\label{eq:haantjes_torsion}
\begin{split}
(\mathcal{H}_L)^i_{jk} = \sum_{\alpha=1}^n \bigg( & -2(L^3)^i_\alpha \partial_{[j} L^\alpha_{k]} + (L^2)^i_\alpha \left( \partial_{[j} (L^2)^\alpha_{k]} + 4 \sum_{\beta=1}^n L^\beta_{[j} \partial_{|\beta|} L^\alpha_{k]} \right) \\
& - 2L^i_\alpha \sum_{\beta=1}^n\left( L^\beta_{[j} \partial_{|\beta|} (L^2)^\alpha_{k]} + (L^2)^\beta_{[j} \partial_{|\beta|} L^\alpha_{k]} \right) + (L^2)^\alpha_{[j} \partial_{|\alpha|} (L^2)^i_{k]} \bigg),
\end{split}
\ee
where indices in square brackets are skew-symmetrized, except those in vertical bars.

An operator $L$ is called a \emph{Haantjes operator} (respectively \emph{Nijenhuis operator}) if $\mathcal{H}_L \equiv 0$ (respectively $T_L \equiv 0$). Every Nijenhuis operator is also Haantjes, but the converse does not hold in general.

\medskip

\noindent \textbf{Symplectic-Haantjes manifolds.}
Let $(M,\omega)$ be a symplectic manifold of dimension $2n$. A \emph{symplectic-Haantjes} (or $\omega\mathscr{H}$) manifold of class $m$ is a triple $(M,\omega,\mathscr{H})$ where:
\begin{itemize}
	\item $\mathscr{H}$ is a Haantjes algebra of rank $m$: a set of Haantjes operators $K : TM \to TM$ closed under $C^\infty(M)$-linear combinations and under composition, generated by $m$ independent operators;
	\item $\omega$ and $\mathscr{H}$ are \emph{algebraically compatible}: for all $K \in \mathscr{H}$,
	\be
	\label{eq:compatibility}
	\Omega K = K^T \Omega,
	\ee
	where $\Omega := \omega^\flat : TM \to T^*M$ is defined by $\omega(X,Y) = \langle \Omega X, Y \rangle$.
\end{itemize}
\noindent While general Haantjes algebras need not be Abelian, by \cite[Proposition~2]{Reyes2024}, compatibility with $\omega$ forces $\mathscr{H}$ to be Abelian: all operators in $\mathscr{H}$ commute pairwise under composition.

\noindent In Darboux (canonical) coordinates $(q,p) = (q^1,\ldots,q^n,p_1,\ldots,p_n)$, condition \eqref{eq:compatibility} implies that $K$ has the block structure
\be
\label{eq:symplectic_structure}
K = \begin{pmatrix} A(q,p) & B(q,p) \\ C(q,p) & A^T(q,p) \end{pmatrix}, \quad B + B^T = 0, \quad C + C^T = 0,
\ee
where $A$, $B$, $C$ are $n \times n$ matrices.

\medskip

\noindent \textbf{Haantjes chains.}
Given a function $H \in C^\infty(M)$ and a distinguished basis $\{K_1,\ldots,K_m\}$ of $\mathscr{H}$, we say that $H$ \emph{generates a Haantjes chain} of length $m$ if
\be
d(K_\alpha^T dH) = 0, \quad \alpha = 1,\ldots,m.
\ee
In practice, we use the local version of the chain equations
\be
\label{eq:chain_equations}
dH_\alpha = K_\alpha^T dH,
\ee
where the (locally) exact $1$-forms $dH_\alpha$ are the \emph{elements of the Haantjes chain}, and their potential functions $H_\alpha$ are in involution with respect to the Poisson bracket induced by $\omega$.

\subsection{The main separation theorem and the computational procedure}
\label{subsec3.2}

The fundamental connection between Haantjes geometry and separation of variables is given by the following result \cite[Theorem 2]{nozaleda2022classical}.

\begin{Theo}[Jacobi--Haantjes]
	\label{thm:jacobi_haantjes}
	Let $(M,\omega,\mathscr{H})$ be an Abelian semisimple $\omega\mathscr{H}$ manifold of class $n$ and $\{H_1=H, H_2, \ldots, H_n\}$ be a set of $C^\infty(M)$ functions belonging to a Haantjes chain generated by $H \in C^\infty(M)$ via a basis $\{K_1=\emph{I},\ldots,K_n\} \subset \mathscr{H}$, where \emph{I} is the identity operator. Then each set $(q,p)$ of Darboux--Haantjes (DH) coordinates provides separation variables for the Hamilton--Jacobi equation associated with each function $H_j$.

	Conversely, if $M$ is a symplectic manifold and $\{H_1, H_2, \ldots, H_n\}$ are $n$ independent $C^\infty(M)$ functions separable in a set of Darboux coordinates $(q,p)$, then they belong to the Haantjes chain generated by the operators
	\be
	\label{eq:separation_operators}
	K_\alpha = \sum_{i=1}^n \frac{\partial H_\alpha}{\partial p_i} \bigg/ \frac{\partial H}{\partial p_i} \left( \frac{\partial}{\partial q^i} \otimes dq^i + \frac{\partial}{\partial p_i} \otimes dp_i \right), \quad \alpha = 1,\ldots,n,
	\ee
	where $H$ is any of the functions $\{H_1,\ldots,H_n\}$ with $\frac{\partial H}{\partial p_i} \neq 0$ for $i=1,\ldots,n$. These operators generate a semisimple $\omega\mathscr{H}$ structure on $M$.
\end{Theo}

A semisimple Abelian $\omega\mathscr{H}$ manifold admits \emph{Darboux--Haantjes} (DH) coordinates in which the symplectic form takes Darboux form $\omega=\sum_{i=1}^n dq^i \wedge dp_i$ and all operators in $\mathscr{H}$ simultaneously diagonalize. These coordinates provide separation variables for the Hamilton--Jacobi equation. For multiseparable systems possessing multiple inequivalent separation coordinate systems, Theorem~\ref{thm:jacobi_haantjes} implies the existence of multiple independent semisimple Abelian $\omega\mathscr{H}$ structures, one for each separation coordinate system.

\medskip

\noindent \textbf{Computational procedure.}\label{algo}
To determine the $\omega\mathscr{H}$ structures for a given integrable Hamiltonian system with Hamiltonian $H$ and integrals of motion $\{H_1=H,\ldots,H_n\}$, we follow a two-step procedure.

\paragraph{Step A: Construction of Haantjes operators.}\label{stepA}
We seek operators $K_\alpha \in \mathscr{H}$ satisfying three conditions:
\begin{enumerate}
	\item \textbf{Haantjes condition:} $\mathcal{H}_{K_\alpha}(X,Y) = 0$ for all $X,Y \in TM$ (using equation \eqref{eq:haantjes_torsion});
	\item \textbf{Compatibility with $\omega$:} $K_\alpha^T \Omega = \Omega K_\alpha$ (giving the structure \eqref{eq:symplectic_structure});
	\item \textbf{Chain equations:} $K_\alpha^T dH = dH_\alpha$ for $\alpha = 1,\ldots,n$ (equation \eqref{eq:chain_equations}).
\end{enumerate}
The construction proceeds as follows. One starts from the block ansatz \eqref{eq:symplectic_structure} in Darboux coordinates and solves the chain equations \eqref{eq:chain_equations}, which constitute a system of linear PDEs in the components of $K_\alpha$. The remaining Haantjes condition \eqref{eq:haantjes_torsion} is a system of nonlinear PDEs and typically requires additional ans\"atze to reduce the number of unknowns.

\paragraph{Step B: Construction of separation coordinates.}\label{stepB}
For a semisimple $\omega\mathscr{H}$ structure, DH coordinates can be constructed systematically. We outline the general procedure from \cite[Section 3.8]{nozaleda2022classical}.

Let $K\in \mathscr{H}$ be a Haantjes operator with pointwise distinct eigenvalues $\{\lambda_1(x),\ldots,\lambda_n(x)\}$. The spectral decomposition of tangent spaces is
\be
T_x M = \bigoplus_{i=1}^n \mathcal{D}_i(x),
\ee
where
\be
\label{eigendist}
\mathcal{D}_i = \ker(K - \lambda_i \text{I}), \quad \mathcal{E}_i = \bigoplus_{j \neq i} \mathcal{D}_j.
\ee
By Haantjes' seminal result \cite{haantjes1955}, the eigendistributions $\mathcal{D}_i$ are \textit{mutually integrable} (any sum $\mathcal{D}_i+\ldots+\mathcal{D}_s$ with all indices different is integrable) and of rank $2$. Their integral leaves are two-dimensional symplectic submanifolds, symplectically orthogonal to each other
\be
\omega(\mathcal{D}_j, \mathcal{D}_k) = 0 \quad \text{for } j \neq k.
\ee
Correspondingly, the cotangent spaces decompose as $T^*_x M = \bigoplus_{i=1}^n \mathcal{E}^{\circ}_i(x)$, where
\be
\label{annih}
\mathcal{E}^{\circ}_i := \ker(K^T - \lambda_i \text{I}) = \Omega(\mathcal{D}_i)
\ee
are the characteristic co-distributions of rank $2$, which are annihilators of the corresponding distributions $\mathcal{E}_i$. The DH coordinates $(x_i, y_i)$ are pairs of characteristic functions that are constant on the distributions $\mathcal{E}_i$ and satisfy canonical conjugacy relations. Their construction proceeds as follows:

\begin{enumerate}
	\item \textbf{Determine a basis of characteristic co-distributions.} For each $i=1,\ldots,n$, find a basis $\{\sigma_i, \tau_i\}$ of $1$-forms for $\mathcal{E}^{\circ}_i$, satisfying
	\be
	(K^T - \lambda_i \text{I})\sigma_i = 0, \quad (K^T - \lambda_i \text{I})\tau_i = 0.
	\ee

	\item \textbf{Search for exact $1$-forms $\alpha_i=dx_i \in \mathcal{E}^{\circ}_i$.} Represent $\alpha_i$ as
	\be
	\label{alpha}
	\alpha_i = f_i \sigma_i + g_i \tau_i,\qquad f_i,g_i \in C^\infty(M),
	\ee
	and require $d\alpha_i = 0$. This determines the integrating factors $f_i$ and $g_i$, often requiring an ansatz for their functional form.

	\item \textbf{Find potential functions $x_i$.} Integrate $\alpha_i = dx_i$ to obtain the coordinates $\{x_1,\ldots,x_n\}$. By construction, these functions automatically satisfy $\{x_i, x_j\} = 0$ for all $i,j$.

	\item \textbf{Search for canonically conjugate momenta $y_i$.} For each $i$, seek another exact $1$-form $\beta_i = dy_i \in \mathcal{E}^{\circ}_i$ linearly independent of $dx_i$. Express it as
	\be
	\label{beta}
	\beta_i = h_i \, dx_i + r_i \, \tau_i,
	\ee
	where $\tau_i$ is assumed linearly independent of $dx_i$.

	\item \textbf{Impose canonical conjugacy.} Require
	\be
	1 = \{x_i, y_i\} = \langle dx_i, P \, dy_i \rangle = h_i \langle dx_i, P \, dx_i \rangle + r_i \langle dx_i, P \, \tau_i \rangle,
	\ee
	where $P = \Omega^{-1}$ is the Poisson bivector. Since $\langle dx_i, P \, dx_i \rangle = 0$ by the antisymmetry of $P$, this yields
	\be
	\label{ri}
	r_i = \frac{1}{\langle dx_i, P \, \tau_i \rangle}.
	\ee

	\item \textbf{Determine the normalizing factor $h_i$.} Substitute $r_i$ into $\beta_i$ \eqref{beta} and impose $d\beta_i = 0$ to find $h_i$.

	\item \textbf{Find the potential $y_i$.} Integrate the exact $1$-form $\beta_i = dy_i$ to obtain $y_i$.
\end{enumerate}
By construction, the coordinates $(x_1,\ldots,x_n,y_1,\ldots,y_n)$ satisfy
\be
\{x_i, x_j\} = 0, \quad \{y_i, y_j\} = 0, \quad \{x_i, y_j\} = \delta_{ij},
\ee
and thus form a Darboux coordinate system. Moreover, all operators $K \in \mathscr{H}$ simultaneously diagonalize in these coordinates, making them DH coordinates and hence separation variables.

\begin{rem}[Semisimple operators with repeated eigenvalues]
\label{rem:semisimple}
The procedure above is stated for simple spectrum (all eigenvalues $\lambda_i$ pointwise distinct),
but it extends without modification to \emph{semisimple} Haantjes operators, i.e.\ those for
which algebraic and geometric multiplicity coincide for every eigenvalue, equivalently, those that
admit a local frame of eigenvectors.
In that case some eigenvalues may coincide, so each eigendistribution
$\mathcal{D}_i = \ker(K-\lambda_i\mathrm{I})$ has rank $\geq 2$, and the co-distribution
$\mathcal{E}^\circ_i$ has correspondingly higher rank.
The Haantjes condition still guarantees that each $\mathcal{D}_i$ is integrable \cite{haantjes1955},
and the construction of exact $1$-forms in steps 1--7 proceeds as before within each eigenspace.
The higher rank simply means there is more freedom in choosing the separation coordinates within
that eigenspace; all choices are equivalent.

All the Haantjes operators we find for the Zernike system in Section~\ref{sec4} have two distinct eigenvalues and the simplified procedure for simple spectrum applies.
\end{rem}

\begin{rem}[Simplification via Nijenhuis generators \cite{nozaleda2022classical}]
	\label{rem:nijenhuis_simplification}
	When the Haantjes algebra admits a Nijenhuis generator $\mathcal{N}$ (which is always possible for semisimple algebras \cite[Proposition~38]{tempesta2022haantjes}), the procedure simplifies significantly. The eigenvalues $\{\lambda_1(q,p),\ldots,\lambda_n(q,p)\}$ of $\mathcal{N}$ are themselves characteristic functions of the Haantjes web and directly provide the coordinates $\{x_1,\ldots,x_n\}$. In this case, one may skip steps 2 and 3, proceeding to determine the conjugate momenta $y_i$ using steps 4--7. This shortened procedure is illustrated in \cite[Section~5]{nozaleda2022classical} for the Drach--Holt system.

	Since the eigenvalues of $\mathcal{N}$ already furnish the new position coordinates $x_i$, one may also forgo the above algorithm and determine the conjugate momenta $y_i$ via a generating function of the canonical transformation.
\end{rem}

\noindent In Sections~\ref{sec4}--\ref{sec:no_ept}, a central role is played by Haantjes operators of \emph{lift form}: operators with $B = 0$ in the block decomposition~\eqref{eq:symplectic_structure}. These are the symplectic analogues of cotangent lifts of Killing tensors on configuration space; they arise when the DH separation coordinates can be reached by an extended point transformation (the EPT forces $v = Jq$ to be aligned with a coordinate axis, as shown in Section~\ref{sec:no_ept}). The following lemma collects their basic structural properties.

\begin{prop}[A-block reduction for lift-form operators]
\label{prop:a_block}
Let $K = \bigl(\begin{smallmatrix} A & 0 \\ C & A^T \end{smallmatrix}\bigr)$ be a Haantjes operator with $B=0$ in the decomposition~\eqref{eq:symplectic_structure}.  Then:
\begin{enumerate}
  \item[(i)] $C$ is skew-symmetric (asserted by~\eqref{eq:symplectic_structure}).
  \item[(ii)] The eigenvalues of $K$ coincide with those of $A$: $\,\det(K-\lambda I)=[\det(A-\lambda I)]^2$.
  \item[(iii)] Any left eigenvector $\sigma$ of $A$ for eigenvalue $\lambda$ gives a position-only left eigenvector $(\sigma,0)$ of $K$ for the same $\lambda$; equivalently, $(\sigma^T,0)^T$ is an eigenvector of $K^T$ for eigenvalue $\lambda$, hence a characteristic $1$-form of $K$ (in the sense of~\eqref{annih}).
\end{enumerate}
Consequently, for a lift-form Haantjes operator the DH coordinates are determined entirely by the $A$-block. Parts~(ii) and~(iii) provide the position coordinates $Q_i = \lambda_i(q)$ via eigenvalues and characteristic $1$-forms of $A$. The conjugate momenta are recovered via the generating function $F_2 = \sum_i P_i Q_i(q)$, which depends on $A$-block data only.
\end{prop}
\begin{proof}
\emph{(i)}~Equation~\eqref{eq:symplectic_structure}.
\emph{(ii)}~Block-triangular determinant: $\det\bigl(\begin{smallmatrix}A{-}\lambda I & 0 \\ C & A^T{-}\lambda I\end{smallmatrix}\bigr)=\det(A{-}\lambda I)\cdot\det(A^T{-}\lambda I)=[\det(A{-}\lambda I)]^2$.
\emph{(iii)}~Direct: $(\sigma,0)K=(\sigma A,0)=(\lambda\sigma,0)$. By transposition, $A^T\sigma^T=\lambda\sigma^T$, and since $K^T=\bigl(\begin{smallmatrix}A^T & C^T \\ 0 & A\end{smallmatrix}\bigr)$, one has $K^T(\sigma^T,0)^T=(A^T\sigma^T,0)^T=\lambda(\sigma^T,0)^T$.
\end{proof}

\section{The separation of variables for the superintegrable Zernike system}
\label{sec4}
\subsection{On the role of the angular momentum \texorpdfstring{$\mathcal{J}$}{C}: polar separation}
\label{subsec4.1}

The angular momentum $\mathcal{J} = q_1 p_2 - q_2 p_1$ is an integral of motion for every Hamiltonian $H_{(N)}$ in the generalized Zernike family, since $H_{(N)}$ depends on $q_1, q_2, p_1, p_2$ only through the rotation-invariant quantities $|\mathbf{p}|^2 = p_1^2 + p_2^2$ and $\mathbf{q} \cdot \mathbf{p} = q_1 p_1 + q_2 p_2$. In polar coordinates, $\mathcal{J}$ coincides with the canonical momentum $p_\varphi$ conjugate to $\varphi$; we use $p_\varphi$ for it in that context. We now apply the computational procedure of Subsection~\ref{subsec3.2} to the pair $(H_{(N)},\frac{1}{2} \mathcal{J}^2)$ and show that it recovers the classical polar separation. (We take $\mathcal{J}^2$ instead of $\mathcal{J}$ for convenience, since the Haantjes method works best with quadratic integrals.)

\noindent \textbf{Step A: The Haantjes operator.}
Solving the chain equation $K_{\mathcal{J}^2}^T \, dH_{(N)} = \mathcal{J}\,d\mathcal{J}$ together with the symplectic compatibility \eqref{eq:symplectic_structure} and the Haantjes condition, we obtain the Haantjes operator associated with $\mathcal{J}^2$ (known from \cite[eq.~(77) with $c=0$]{nozaleda2022classical})
\be
\label{eq:haantjes_pphi}
K_{\mathcal{J}^2} = \begin{pmatrix}
q_2^2 & -q_1 q_2 & 0 & 0 \\[0.3em]
-q_1 q_2 & q_1^2 & 0 & 0 \\[0.3em]
0 & -(q_1 p_2 - q_2 p_1) & q_2^2 & -q_1 q_2 \\[0.3em]
(q_1 p_2 - q_2 p_1) & 0 & -q_1 q_2 & q_1^2
\end{pmatrix}.
\ee
This operator is semisimple, with eigenvalues
\be
\label{eq:eigenvalues_pphi}
\lambda_1 = q_1^2 + q_2^2, \qquad \lambda_2 = 0,
\ee
each of multiplicity two. A notable feature is that $K_{\mathcal{J}^2}$ does not depend on any of the parameters $\gamma_n$, reflecting the universal character of the angular momentum as an integral of all generalized Zernike Hamiltonians.

\medskip

\noindent \textbf{Step B: DH coordinates.}
The characteristic co-distributions \eqref{annih} of $K_{\mathcal{J}^2}$ are
\be
\label{eq:annih_pphi}
\begin{split}
\mathcal{E}_1^{\circ} &= \Big[ (q_1 p_2 - q_2 p_1) \, dq_1 - q_1 q_2 \, dp_1 + q_1^2 \, dp_2, \;\; -q_2 \, dq_1 + q_1 \, dq_2 \Big], \\[0.5em]
\mathcal{E}_2^{\circ} &= \Big[ {-(q_1 p_2 - q_2 p_1)} \, dq_1 + q_1 q_2 \, dp_1 + q_2^2 \, dp_2, \;\; q_1 \, dq_1 + q_2 \, dq_2 \Big].
\end{split}
\ee

\paragraph{Position coordinates.} Each co-distribution contains a $1$-form involving only the position differentials $dq_1, dq_2$. From $\mathcal{E}_1^{\circ}$, we integrate
\be
-q_2 \, dq_1 + q_1 \, dq_2 = q_1^2 \, d\!\left(\frac{q_2}{q_1}\right),
\ee
obtaining the ratio $q_2/q_1$. Since any function of $q_2/q_1$ is an equally valid characteristic coordinate, the natural choice is the polar angle $\varphi = \arctan(q_2/q_1)$. From $\mathcal{E}_2^{\circ}$, we have
\be
q_1 \, dq_1 + q_2 \, dq_2 = \tfrac{1}{2}\, d (q_1^2 + q_2^2)=\tfrac{1}{2}\, d r^2, \qquad r := \sqrt{q_1^2 + q_2^2},
\ee
giving the radial coordinate.

\paragraph{Conjugate momenta.} We seek an exact $1$-form $dp_\varphi \in \mathcal{E}_1^\circ$.  Setting
\be
dp_\varphi = \frac{1}{q_1}\bigl((q_1 p_2 - q_2 p_1)\,dq_1 - q_1 q_2\,dp_1 + q_1^2\,dp_2\bigr)
            -\frac{p_1}{q_1}\bigl(-q_2\,dq_1 + q_1\,dq_2\bigr),
\ee
one checks directly that this equals $p_2\,dq_1 - p_1\,dq_2 + q_1\,dp_2 - q_2\,dp_1 = d(q_1 p_2 - q_2 p_1)$, confirming $p_\varphi = q_1 p_2 - q_2 p_1$.  Similarly, from $\mathcal{E}_2^\circ$ one obtains $p_r = (q_1 p_1 + q_2 p_2)/\sqrt{q_1^2+q_2^2}$.  One verifies $\{r, p_r\} = \{\varphi, p_\varphi\} = 1$ and $\{r, p_\varphi\} = \{\varphi, p_r\} = 0$.

\paragraph{The canonical transformation.} Collecting the results, the DH coordinates for the pair $(H_{(N)}, \tfrac{1}{2}\mathcal{J}^2)$ are the standard polar coordinates on $T^*\mathbb{R}^2$:
\be
\label{eq:canonical_polar}
\boxed{
\begin{aligned}
r &= \sqrt{q_1^2 + q_2^2}, &\qquad p_r &= \frac{q_1 p_1 + q_2 p_2}{\sqrt{q_1^2 + q_2^2}}, \\[0.5em]
\varphi &= \arctan \frac{q_2}{q_1}, &\qquad p_\varphi &= q_1 p_2 - q_2 p_1.
\end{aligned}
}
\ee

\paragraph{Separated Hamiltonian.} Using $|\mathbf{p}|^2 = p_r^2 + p_\varphi^2/r^2$ and $\mathbf{q} \cdot \mathbf{p} = r \, p_r$, the Hamiltonian in polar coordinates reads
\be
\label{eq:HN_polar}
H_{(N)} = p_r^2 + \frac{p_\varphi^2}{r^2} + \sum_{n=1}^{N} \gamma_n \, (r \, p_r)^n,
\ee
and the integral of motion is simply $p_\varphi$ itself. The Hamilton--Jacobi equation separates upon setting $p_\varphi = \mathrm{const}$, reducing to a single ODE in the radial variable~$r$. This is the well-known polar separation of the Zernike system, here recovered by purely algebraic means from the Haantjes framework. We note that, unlike the DH coordinates of Subsection~\ref{subsec4.2}, the polar transformation \eqref{eq:canonical_polar} does not depend on any parameter $\gamma_n$ and provides separation for the entire family $\{H_{(N)}\}_{N \geq 1}$ simultaneously.

\begin{rem}[Curvature-blindness of the Haantjes construction]
\label{rem:curvature_blind}
The Haantjes operator $K_{\mathcal{J}^2}$ \eqref{eq:haantjes_pphi}, produced by \hyperref[stepA]{Step~A}, does not depend on $\gamma_2$, nor on any other parameter $\gamma_n$. The entire \hyperref[stepB]{Step~B} computation is therefore independent of the curvature $\kappa = -\gamma_2$ of the underlying space. In particular, the characteristic co-distributions \eqref{eq:annih_pphi} and all subsequent integration steps are the same for $\mathbb{E}^2$, $\mathbb{S}^2$, and $\mathbb{H}^2$.

As a consequence, the Haantjes method does not intrinsically select between the Euclidean radial coordinate $r = \sqrt{q_1^2 + q_2^2}$ and the geodesic distance $\rho$ related by $r = S_\kappa(\rho)$: both are admissible characteristic coordinates for the same co-distribution $\mathcal{E}_2^\circ$, since any strictly monotone function of $r$ is equally valid. The geometric distinction between flat polar coordinates $(r, \varphi)$ and geodesic polar coordinates $(\rho, \varphi)$ lies outside the reach of the Haantjes algebra.
\end{rem}

\paragraph{Integrals $I_1$, $I_2$ in polar coordinates, $N=2$.} While the polar Hamiltonian \eqref{eq:HN_polar} involves only $r$ and $p_r$ (for fixed $p_\varphi$), the remaining integrals $I_1$, $I_2$ \eqref{eq:addconst} take a natural factored form in polar coordinates. Substituting $q_1 = r\cos\varphi$, $q_2 = r\sin\varphi$ and the conjugate momenta \eqref{polcoor} into \eqref{eq:addconst}, one obtains
\begin{align}
I_1 &= (1+\gamma_2 r^2)\!\left(\cos\varphi\, p_r - \frac{\sin\varphi}{r}\, p_\varphi\right)^{\!2} + \gamma_1 r\cos\varphi\!\left(\cos\varphi\, p_r - \frac{\sin\varphi}{r}\, p_\varphi\right), \label{eq:I1_polar} \\[0.5em]
I_2 &= (1+\gamma_2 r^2)\!\left(\sin\varphi\, p_r + \frac{\cos\varphi}{r}\, p_\varphi\right)^{\!2} + \gamma_1 r\sin\varphi\!\left(\sin\varphi\, p_r + \frac{\cos\varphi}{r}\, p_\varphi\right). \label{eq:I2_polar}
\end{align}
These expressions simplify considerably in the geodesic polar coordinates $(\rho, \varphi, p_\rho, p_\varphi)$ of Subsection~\ref{sec2.2}, in which $H_{(2)}$ takes the natural form \eqref{eq:hamigeorada}. Using $r = S_\kappa(\rho)$, $1 + \gamma_2 r^2 = C_\kappa^2(\rho)$ and the canonical transformation therein, the integrals become
\be
\label{eq:I1_geodpolar}
I_1 = \left(\cos\varphi \, p_\rho - \frac{\sin\varphi}{T_\kappa(\rho)}\, p_\varphi\right)^{\!2} - \frac{\gamma_1^2}{4}\,T_\kappa^2(\rho)\cos^2\!\varphi \,,
\ee
\be
\label{eq:I2_geodpolar}
I_2 = \left(\sin\varphi \, p_\rho + \frac{\cos\varphi}{T_\kappa(\rho)}\, p_\varphi\right)^{\!2} - \frac{\gamma_1^2}{4}\,T_\kappa^2(\rho)\sin^2\!\varphi \,.
\ee
In this form, the structure is transparent: each integral is a ``partial energy'', decomposed into a kinetic contribution (the squared term) and a potential contribution $-\frac{\gamma_1^2}{4}T_\kappa^2(\rho)$ weighted by $\cos^2\varphi$ or $\sin^2\varphi$ respectively. In particular, $\gamma_2$ enters only through $\kappa = -\gamma_2$ in the curvature-dependent functions $S_\kappa$, $C_\kappa$, $T_\kappa$ \eqref{curvsine}--\eqref{ktan}. The consistency relation \eqref{eq:dep}, namely $I_1 + I_2 - \gamma_2\, p_\varphi^2 = H_{(2)}$, follows from the identity $1/T_\kappa^2(\rho) + \kappa = 1/S_\kappa^2(\rho)$.

\subsection{Cartesian-type separation associated with \texorpdfstring{$I_2$}{I2}}
\label{subsec4.2}

We now apply the computational procedure of Subsection~\ref{subsec3.2} to construct DH coordinates associated with the integral $I_2$. We work in Cartesian coordinates $(q_1, q_2, p_1, p_2)$ on the phase space $T^*\mathbb{R}^2$.

\medskip

\noindent \textbf{Step A: The Haantjes operator.}
Using a polynomial ansatz for the entries of the block structure \eqref{eq:symplectic_structure} and solving the chain equation $K_{I_2}^T \, dH_{(N)} = dI_2$ together with the Haantjes condition, we obtain the Haantjes operator $K_{I_2}$. For $N \leq 2$, the result is
\be
\label{eq:haantjes_I2}
K_{I_2} = \begin{pmatrix}
0 & 0 & 0 & 0 \\[0.3em]
-\gamma_2 q_1 q_2 & 1+ \gamma_2 q_1^2 & 0 & 0 \\[0.3em]
0 & -\gamma_2 q_1 p_2 & 0 & -\gamma_2 q_1 q_2 \\[0.3em]
\gamma_2 q_1 p_2 & 0 & 0 & 1+ \gamma_2 q_1^2
\end{pmatrix},
\ee
with eigenvalues $\lambda_1 = 0$ and $\lambda_2 = 1+ \gamma_2 q_1^2$, each of multiplicity two. This operator is semisimple, confirming the applicability of Theorem~\ref{thm:jacobi_haantjes}. Setting $\gamma_2 = 0$ recovers the $N = 1$ case
\be
\label{eq:haantjes_I1}
K_{I_2}\big|_{\gamma_2=0} = \operatorname{diag}(0, 1, 0, 1).
\ee

\medskip

\noindent \textbf{Nijenhuis generator.}
For systems with two degrees of freedom, we consider the operator \cite[Remarks~3 and~4]{nozaleda2022classical}
\be
\label{eq:nijenhuis_generator}
\mathcal{N}_\alpha = K_\alpha - \tfrac{1}{2} \operatorname{tr}(K_\alpha) \, I.
\ee
If $\mathcal{N}_\alpha$ is a Nijenhuis operator, then the Hamiltonian system admits a quasi-bi-Hamiltonian formulation \cite{morosi1997,morosi1998,tondo1999}, and the eigenvalues of $\mathcal{N}_\alpha$ are characteristic functions of the Haantjes web, i.e., they directly provide the new position coordinates. Applying \eqref{eq:nijenhuis_generator} to \eqref{eq:haantjes_I2}, we obtain
\be
\label{eq:nijenhuis_I2}
\mathcal{N}_{I_2} = \begin{pmatrix}
-(1+ \gamma_2 q_1^2) & 0 & 0 & 0 \\[0.3em]
-\gamma_2 q_1 q_2 & 0 & 0 & 0 \\[0.3em]
0 & -\gamma_2 q_1 p_2 & -(1+ \gamma_2 q_1^2) & -\gamma_2 q_1 q_2 \\[0.3em]
\gamma_2 q_1 p_2 & 0 & 0 & 0
\end{pmatrix},
\ee
which we have verified to be a Nijenhuis operator for $N \leq 2$, with eigenvalues
\be
\label{eq:eigenvalues_N2}
\lambda_1 = -(1+ \gamma_2 q_1^2), \qquad \lambda_2 = 0.
\ee
A key observation is that these eigenvalues depend only on the coordinate $q_1$ and not on the momenta. This means that the new position coordinates $(Q_1,Q_2)$ will be functions of $(q_1, q_2)$ alone---an extended point transformation will be sufficient for separation. As we show in Section~\ref{sec:no_ept}, this is a special feature of $N \leq 2$ that does not persist for $N \geq 3$.

\medskip

\noindent \textbf{The case $\boldsymbol{N = 1}$.}
Setting $\gamma_2 = 0$, the Haantjes operator \eqref{eq:haantjes_I1} is already diagonal in Cartesian coordinates, and the Nijenhuis generator reduces to $\mathcal{N}_{I_2}\big|_{\gamma_2=0} = \operatorname{diag}(-1, 0, -1, 0)$. This means that $(q_1, q_2, p_1, p_2)$ are themselves DH coordinates for the pair $(H_{(1)}, I_2)$, and the Hamiltonian and integral are already in separated form
\be
\label{eq:H1_separated}
H_{(1)} = \underbrace{(p_1^2 + \gamma_1 q_1 p_1)}_{I_1} + \underbrace{(p_2^2 + \gamma_1 q_2 p_2)}_{ I_2}.
\ee
The integrals in these DH coordinates are
\be
\label{eq:I1_separated}
I_1 = p_1^2 + \gamma_1 q_1 p_1, \qquad  I_2 = p_2^2 + \gamma_1 q_2 p_2.
\ee
By the $q_1 \leftrightarrow q_2$ symmetry of $H_{(1)}$, the companion integral $I_1$ admits the Haantjes operator $K_{I_1} = \operatorname{diag}(1, 0, 1, 0)$. Since $K_{I_2}\big|_{\gamma_2=0} + K_{I_1} = I$, this is consistent with the relation $H_{(1)} = I_1 + I_2$.

We also note that the canonical transformation $Q_i = q_i$, $P_i = p_i + \frac{\gamma_1}{2} q_i$ brings $H_{(1)}$ into the isotropic harmonic oscillator form
\be
\label{eq:H1_oscillator}
H_{(1)} = P_1^2 + P_2^2 - \frac{\gamma_1^2}{4}(Q_1^2 + Q_2^2),
\ee
with integrals
\be
\label{eq:I1_oscillator}
I_2 = P_2^2 - \frac{\gamma_1^2}{4} Q_2^2, \qquad I_1 = P_1^2 - \frac{\gamma_1^2}{4} Q_1^2.
\ee
For $\gamma_1 = \imath\omega$ with $\omega \in \mathbb{R}$, these reduce to the standard harmonic oscillator Hamiltonian and its partial energy integrals.

\medskip

\noindent \textbf{The case $\boldsymbol{N = 2}$:}
The eigenvalue $\lambda_2 = 0$ of the Nijenhuis operator \eqref{eq:nijenhuis_I2} cannot serve as a separation coordinate (a constant function cannot parametrise a coordinate chart), so the Nijenhuis shortcut of Remark~\ref{rem:nijenhuis_simplification} provides only one useful coordinate. We therefore apply the full \hyperref[stepB]{Step~B} algorithm.

The characteristic co-distributions \eqref{annih} for the operator $\mathcal{N}_{I_2}$ are
\be
\label{eq:annih_N2}
\begin{split}
\mathcal{E}_1^{\circ} &= \Big[ \gamma_2 q_1 p_2 \, dq_2 + (1 + \gamma_2 q_1^2) \, dp_1 + \gamma_2 q_1 q_2 \, dp_2, \;\; dq_1 \Big], \\[0.5em]
\mathcal{E}_2^{\circ} &= \Big[ \gamma_2 q_1 p_2 \, dq_1 + (1 + \gamma_2 q_1^2) \, dp_2, \;\; -\gamma_2 q_1 q_2 \, dq_1 + (1 + \gamma_2 q_1^2) \, dq_2 \Big].
\end{split}
\ee
A crucial feature of \eqref{eq:annih_N2} is that $\mathcal{E}_2^{\circ}$ contains a $1$-form involving only the coordinates $q_1, q_2$ (the second basis element), while $\mathcal{E}_1^{\circ}$ contains the exact form $dq_1$. This confirms that we can obtain the DH coordinates by an extended point transformation.

\paragraph{Position coordinates.} From $\mathcal{E}_1^{\circ}$ we immediately read off
\be
Q_1 = q_1.
\ee
From the second basis element of $\mathcal{E}_2^{\circ}$, we integrate
\be
\frac{1}{(1+\gamma_2 q_1^2)\,q_2}\Big(-\gamma_2 q_1 q_2\,dq_1 + (1+\gamma_2 q_1^2)\,dq_2\Big) = d\!\left(\ln q_2 - \tfrac{1}{2}\ln(1+\gamma_2 q_1^2)\right),
\ee
and, exponentiating (any strictly monotone function of a characteristic coordinate is again a valid coordinate), obtain
\be
\label{eq:Q2}
Q_2 = \frac{q_2}{\sqrt{1 + \gamma_2 q_1^2}}.
\ee

\paragraph{Conjugate momenta.} From $\mathcal{E}_2^{\circ}$, the first basis element $\gamma_2 q_1 p_2 \, dq_1 + (1 + \gamma_2 q_1^2) \, dp_2$ is proportional to $d\!\big(p_2 \sqrt{1 + \gamma_2 q_1^2}\big)$, giving
\be
\label{eq:P2}
P_2 = p_2 \sqrt{1 + \gamma_2 q_1^2},
\ee
which is already canonically normalized: $\{Q_2, P_2\} = 1$.

For $P_1$, we require $\{Q_1, P_1\} = 1$, which forces $\partial P_1 / \partial p_1 = 1$, so $P_1 = p_1 + f(q_1, q_2, p_2)$ for some function $f$. Demanding that $dP_1 \in \mathcal{E}_1^{\circ}$ and exactness of $dP_1$ yields
\be
\label{eq:P1}
P_1 = p_1 + \frac{\gamma_2 q_1 q_2 \, p_2}{1 + \gamma_2 q_1^2}.
\ee

\paragraph{The canonical transformation.} Collecting the results, the DH coordinates for the pair $(H_{(2)}, I_2)$ are given by the extended point transformation
\be
\label{eq:canonical_N2}
\boxed{
\begin{aligned}
Q_1 &= q_1, &\qquad P_1 &= p_1 + \frac{\gamma_2 q_1 q_2 \, p_2}{1 + \gamma_2 q_1^2}, \\[0.5em]
Q_2 &= \frac{q_2}{\sqrt{1 + \gamma_2 q_1^2}}, &\qquad P_2 &= p_2 \sqrt{1 + \gamma_2 q_1^2}.
\end{aligned}
}
\ee
One verifies that this is a canonical transformation: $\{Q_i, Q_j\} = \{P_i, P_j\} = 0$ and $\{Q_i, P_j\} = \delta_{ij}$. In the limit $\gamma_2 \to 0$, the transformation \eqref{eq:canonical_N2} reduces to the identity, recovering the Cartesian DH coordinates of the $N = 1$ case.

\begin{rem}[Geometric interpretation of the DH coordinates]
\label{rem:geometric_interpretation}
The canonical transformation \eqref{eq:canonical_N2} admits a natural geometric interpretation in terms of the coordinate systems introduced by Pogosyan, Wolf, and Yakhno~\cite{PWY2017} for the separation of the Hamilton--Jacobi equation on the sphere and the hyperboloid.

\smallskip
\noindent\emph{Case $\gamma_2 < 0$ (sphere).}
The Zernike system is defined on the disk $q_1^2 + q_2^2 < R^2$, $R^2 = 1/|\gamma_2|$, which is the orthogonal or vertical projection of the hemisphere $\xi_1^2 + \xi_2^2 + \xi_3^2 = R^2$, $\xi_3 \geq 0$, via $\xi_1 = q_1$, $\xi_2 = q_2$ and $\xi_3 = \sqrt{R^2-q_1^2-q_2^2}$. In the System~II spherical coordinates of~\cite{PWY2017},%
\be
\xi_1 = R\cos\vartheta, \quad \xi_2 = R\sin\vartheta\cos\varphi, \quad \xi_3 = R\sin\vartheta\sin\varphi,
\ee
the key identity $1 + \gamma_2 q_1^2 = 1 - \cos^2\vartheta = \sin^2\vartheta$ gives (choosing the positive root in $Q_2$)
\be
Q_1 = R\cos\vartheta, \qquad Q_2 = \frac{R\sin\vartheta\cos\varphi}{\sin\vartheta} = R\cos\varphi.
\ee
Thus $Q_1$ determines the polar angle~$\vartheta$ and $Q_2$ determines the azimuthal angle~$\varphi$, while $P_2 = p_2\sin\vartheta$ is the momentum conjugate to the projected azimuthal coordinate. The correction term in $P_1$ ensures canonicity of the full transformation, absorbing the cross-term from the $(\mathbf{q}\cdot\mathbf{p})^2$ coupling.

\smallskip
\noindent\emph{Case $\gamma_2 > 0$ (hyperboloid).}
When $\gamma_2 > 0$, the system is defined on the full plane $\mathbb{R}^2$, which is the vertical projection of the upper sheet of the two-sheeted hyperboloid $\xi_3^2 - \xi_1^2 - \xi_2^2 = \varrho^2$, $\varrho^2 = 1/\gamma_2$. In the System~HII equidistant coordinates of~\cite{PWY2017},%
\be
\xi_1 = \varrho\sinh\tau_1, \quad \xi_2 = \varrho\cosh\tau_1\sinh\tau_2, \quad \xi_3 = \varrho\cosh\tau_1\cosh\tau_2,
\ee
one has $1 + \gamma_2 q_1^2 = \cosh^2\tau_1$ and the analogous identification $Q_1 = \varrho\sinh\tau_1$, $Q_2 = \varrho\sinh\tau_2$ holds.

\smallskip
Moreover, the integral $I_2$ in original coordinates,
\be
I_2 = \big[1 + \gamma_2(q_1^2 + q_2^2)\big]\,p_2^2 + \gamma_1\,q_2\,p_2,
\ee
coincides with the separation constant $K_{II}^2$ from~\cite[Eq.~(59)]{PWY2017} (under the identification $\gamma_2 \leftrightarrow \alpha$, $\gamma_1 \leftrightarrow -i\beta$). The Haantjes construction thus recovers the System~II/HII separation of the Hamilton--Jacobi equation by purely algebraic means.
\end{rem}

\paragraph{Separated Hamiltonian and integral.} Expressing $H_{(2)}$ in the DH coordinates \eqref{eq:canonical_N2}, one obtains
\be
\label{eq:H2_separated}
H_{(2)} = P_1^2 + \gamma_1 Q_1 P_1 + \gamma_2 (Q_1 P_1)^2 + \frac{P_2^2 + \gamma_1 Q_2 P_2 + \gamma_2 (Q_2 P_2)^2}{1 + \gamma_2 Q_1^2},
\ee
where the integral of motion takes the manifestly separated form
\be
\label{eq:I2_separated}
I_2 = P_2^2 + \gamma_1 Q_2 P_2 + \gamma_2 (Q_2 P_2)^2.
\ee
Integral $I_1$ has a more complicated, nonseparable expression in these coordinates that we omit.

Thus, $H_{(2)}$ separates as a sum of two contributions, one depending on $(Q_1, P_1)$ and the other on $(Q_2, P_2)$ divided by a conformal factor $1 + \gamma_2 Q_1^2$. This generalizes the $N = 1$ pattern \eqref{eq:H1_separated}: setting $\gamma_2 = 0$ in \eqref{eq:H2_separated} recovers $H_{(1)} = (P_1^2 + \gamma_1 Q_1 P_1) + (P_2^2 + \gamma_1 Q_2 P_2)$, with $Q_i = q_i$ and $P_i = p_i$.

\subsection{Quadratic integrals of motion \texorpdfstring{$I_1$, $I_2$}{I1, I2} and corresponding Cartesian-type separation of variables for \texorpdfstring{$H_{(2)}$}{H(2)}}
\label{subsec4.3}

The Hamiltonian $H_{(2)}$ admits a second Cartesian-type separation, associated with the integral $I_1$ in place of $I_2$. The analysis is entirely parallel to Subsection~\ref{subsec4.2}: the map $(q_1, q_2, p_1, p_2) \mapsto (q_2, q_1, p_2, p_1)$ is a symmetry of $H_{(2)}$ that interchanges $I_1 \leftrightarrow I_2$, and applying it to all constructions of Subsection~\ref{subsec4.2} yields the results for $I_1$ without further computation.

\medskip

\noindent\textbf{Step A: The Haantjes operator.} The Haantjes operator for the pair $(H_{(2)}, I_1)$ is
\be
\label{eq:haantjes_I1_cart}
K_{I_1} = \begin{pmatrix}
1+ \gamma_2 q_2^2 & -\gamma_2 q_1 q_2 & 0 & 0 \\[0.3em]
0 & 0 & 0 & 0 \\[0.3em]
0 & \gamma_2 q_2 p_1 & 1+ \gamma_2 q_2^2 & 0 \\[0.3em]
-\gamma_2 q_2 p_1 & 0 & -\gamma_2 q_1 q_2 & 0
\end{pmatrix},
\ee
with eigenvalues $\lambda_1 = 1 + \gamma_2 q_2^2$ (double) and $\lambda_2 = 0$ (double). Applying \eqref{eq:nijenhuis_generator} gives the Nijenhuis generator
\be
\label{eq:nijenhuis_I1}
\mathcal{N}_{I_1} = \begin{pmatrix}
0 & -\gamma_2 q_1 q_2 & 0 & 0 \\[0.3em]
0 & -(1+ \gamma_2 q_2^2) & 0 & 0 \\[0.3em]
0 & \gamma_2 q_2 p_1 & 0 & 0 \\[0.3em]
-\gamma_2 q_2 p_1 & 0 & -\gamma_2 q_1 q_2 & -(1+ \gamma_2 q_2^2)
\end{pmatrix},
\ee
with eigenvalues $\mu_1 = 0$ and $\mu_2 = -(1+ \gamma_2 q_2^2)$.

\medskip

\noindent\textbf{Step B: The canonical transformation.}
Applying the algorithm of Subsection~\ref{subsec3.2} to the characteristic co-distributions of $\mathcal{N}_{I_1}$ gives the DH coordinates for the pair $(H_{(2)}, I_1)$:
\be
\label{eq:canonical_I1}
\boxed{
\begin{aligned}
Q_1 &= \frac{q_1}{\sqrt{1 + \gamma_2 q_2^2}}, &\qquad P_1 &= p_1 \sqrt{1 + \gamma_2 q_2^2}, \\[0.5em]
Q_2 &= q_2, &\qquad P_2 &= p_2 + \frac{\gamma_2 q_1 q_2 \, p_1}{1 + \gamma_2 q_2^2}.
\end{aligned}
}
\ee
One verifies that \eqref{eq:canonical_I1} is canonical. In the limit $\gamma_2 \to 0$ it reduces to the identity, recovering the Cartesian DH coordinates of the $N = 1$ case.

\paragraph{Separated Hamiltonian and integral.} In the DH coordinates \eqref{eq:canonical_I1},
\be
\label{eq:I1_separated_cart}
I_1 = P_1^2 + \gamma_1 Q_1 P_1 + \gamma_2 (Q_1 P_1)^2,
\ee
\be
\label{eq:H2_separated_I1}
H_{(2)} = \frac{P_1^2 + \gamma_1 Q_1 P_1 + \gamma_2 (Q_1 P_1)^2}{1 + \gamma_2 Q_2^2} + P_2^2 + \gamma_1 Q_2 P_2 + \gamma_2 (Q_2 P_2)^2.
\ee
The structure mirrors \eqref{eq:H2_separated}: now $I_1$ is the separated ``first factor'' and appears divided by the conformal factor $1 + \gamma_2 Q_2^2$ in the Hamiltonian. Comparing \eqref{eq:H2_separated} and \eqref{eq:H2_separated_I1}, the two canonical transformations \eqref{eq:canonical_N2} and \eqref{eq:canonical_I1} are related by $q_1 \leftrightarrow q_2$, $p_1 \leftrightarrow p_2$, confirming the $\mathbb{Z}_2$ symmetry of the problem.

\begin{rem}
\label{rem:pattern}
Comparing the separated forms \eqref{eq:H1_separated}, \eqref{eq:I2_separated}, and \eqref{eq:I1_separated_cart}, a pattern emerges: in DH coordinates adapted to the integral $I_k$ ($k = 1, 2$), the separated form of $I_k$ is
\be
\label{eq:IN_pattern}
I_k = P_k^2+\sum_{n=1}^{N} \gamma_n (Q_k P_k)^n,
\ee
while the Hamiltonian carries a conformal factor $1 + \gamma_2 Q_{3-k}^2$ in the remaining term. We conjecture that this form persists to $N\geq 3$. However, by Theorem~\ref{thm:no_ept}, the DH coordinates for $I_k$ cannot be reached by an EPT when $N\geq 3$; therefore the transformation $(q,p)\to(Q,P)$ realising the conjectured form would have to be momentum-dependent.
\end{rem}

\subsection{Elliptic separation: algebraic DH coordinates and Heun equations}
\label{subsec:elliptic}

Fix an interfocal parameter $f\in(0,\pi/2)$ and set $k_1 := \cos f$, $k_2 := \sin f$.
The separation constant for elliptic coordinates on $\mathbb{S}^2$ is \cite[eq.~(72)]{PWY2017}
\be
  I_e := -\gamma_2 k_1^2\,\mathcal{J}^2 - k_2^2\,I_2.
\ee
By linearity of the chain equation \eqref{eq:chain_equations}, the operator
\begin{equation}\label{eq:K_e_def}
  K_e := -\gamma_2 k_1^2\,K_{\mathcal{J}^2} - k_2^2\,K_{I_2}
\end{equation}
satisfies $K_e^T\,dH_{(2)} = dI_e$.  That $K_e$ is itself a Haantjes operator is not automatic from the linearity of the chain equation (the Haantjes condition is nonlinear in $K$), but follows from a direct computation

\begin{prop}[Elliptic Haantjes operator]
\label{prop:Ke_haantjes}
The operator $K_e = -\gamma_2 k_1^2\,K_{\mathcal{J}^2} - k_2^2\,K_{I_2}$ is Haantjes, i.e., $\mathcal{H}_{K_e}\equiv 0$.  Together with the identity operator it generates a semisimple Abelian $\omega\mathscr{H}$ structure for the pair $(H_{(2)}, I_e)$.
\begin{proof}
Direct computation using the explicit entries of $K_e$ (obtained by substituting the formulas for $K_{\mathcal{J}^2}$ and $K_{I_2}$) shows that all components of the Haantjes torsion \eqref{eq:haantjes_torsion} vanish identically.  The operator is semisimple because its eigenvalues $\tilde\lambda_1 \neq \tilde\lambda_2$ are distinct on the open dense set $\{S \neq 0\}$.
\end{proof}
\end{prop}

\noindent Since a linear combination of lift-form operators is again of lift form ($B=0$), Proposition~\ref{prop:a_block} applies: the DH coordinates are determined by the $A$-block alone and can be reached by an EPT.

\paragraph{The $A$-block and characteristic polynomial.}
\begin{equation}\label{eq:Ae_charpoly}
  A_e = \begin{pmatrix}
    -\gamma_2 k_1^2 q_2^2 & \gamma_2 k_1^2 q_1 q_2 \\[4pt]
    \gamma_2 q_1 q_2 & -\gamma_2 q_1^2 - k_2^2
  \end{pmatrix}, \qquad
  \lambda^2 + T\lambda + P = 0,
\end{equation}
with
\begin{align}
\label{eq:TPS_def}
  T &:= \gamma_2(q_1^2 + k_1^2 q_2^2) + k_2^2, \notag\\
  P &:= \gamma_2 k_1^2 k_2^2 q_2^2, \notag\\
  S &:= \sqrt{T^2 - 4P}.
\end{align}

\paragraph{DH coordinates.}
The eigenvalues of $A_e$ are $\lambda_{1,2} = \tfrac{1}{2}(-T \pm S)$.  The Nijenhuis operator $\mathcal{N}_e = K_e - \tfrac{1}{2}\operatorname{tr}(K_e)\,I$ has eigenvalues $\tilde\lambda_{1,2} = -\lambda_{2,1} = \tfrac{1}{2}(T \pm S)$, which by Remark~\ref{rem:nijenhuis_simplification} are the DH coordinates
\begin{equation}\label{eq:DH_elliptic}
  Q_{1,2} = \tilde\lambda_{1,2}(q_1,q_2) = \tfrac{1}{2}(T \pm S), \qquad \tilde\lambda_1 > \tilde\lambda_2.
\end{equation}
Unlike the Cartesian case (where $Q_1 = q_1$ is rational), both $\tilde\lambda_{1,2}$ depend on both variables through $S$, making them \emph{algebraic of degree~$2$} but not rational.

\paragraph{New momenta.}
Because $K_e$ induces an EPT, the canonical transformation has generating function $F_2(q,P) = P_1 Q_1(q) + P_2 Q_2(q)$, giving $\mathbf{p} = J^T\mathbf{P}$ with $J_{ij} = \partial\tilde\lambda_i/\partial q_j$.  Hence $\mathbf{P} = (J^{-1})^T\mathbf{p}$, and the new momenta are \emph{linear} in the old momenta $p_j$ with algebraic coefficients in $q$.  The Jacobian $J$ is invertible wherever $\tilde\lambda_1 \neq \tilde\lambda_2$, i.e.\ wherever $S = \sqrt{T^2-4P} \neq 0$; the locus $S=0$ is the set of focal points of the confocal family~\eqref{eq:confocal_level_sets_paper}, where the two families of conics become tangent.
The square root is always real: $T^2-4P \geq 0$ throughout $\mathbb{R}^2$ for any sign of $\gamma_2$.  For $\gamma_2 \leq 0$ this is immediate since $P = \gamma_2 k_1^2 k_2^2 q_2^2 \leq 0$.  For $\gamma_2 > 0$ the discriminant factors as
\begin{equation}\label{eq:discriminant_factored}
  T^2-4P
  = \bigl[\gamma_2 q_1^2 + (\sqrt{\gamma_2}\,k_1 q_2+k_2)^2\bigr]
    \bigl[\gamma_2 q_1^2 + (\sqrt{\gamma_2}\,k_1 q_2-k_2)^2\bigr] \geq 0,
\end{equation}
vanishing only at the two isolated focal points $(q_1,q_2)=(0,\pm k_2/(\sqrt{\gamma_2}\,k_1))$.  Hence the elliptic DH coordinates are well-defined on all of $\mathbb{R}^2$ minus the focal locus, with no additional domain restriction required for the hyperboloid case.

\paragraph{Characteristic forms.}
By Proposition~\ref{prop:a_block}(iii), the left eigenvectors of $A_e$ give characteristic $1$-forms of $K_e$.  Solving $\sigma A_e = \lambda_i\,\sigma$ with the normalisation $\sigma_1 = \gamma_2 q_1 q_2$ yields
\begin{equation}\label{eq:sigma_elliptic}
  \sigma_i = \gamma_2 q_1 q_2\,dq_1 + (\gamma_2 k_1^2 q_2^2 + \lambda_i)\,dq_2.
\end{equation}
Note that $\sigma_i$ contains $\lambda_i$, and hence $S = \sqrt{T^2-4P}$, in the $dq_2$-coefficient. Direct integration is therefore not feasible.

\begin{rem}[The square root cannot be removed]
\label{rem:sqrt}
In the polar and Cartesian-typew cases the DH coordinates involve a square root of a single-variable rational function (e.g.\ $Q_1 = q_1/\sqrt{1+\gamma_2 q_2^2}$ for $I_1$), which can be eliminated by reparametrisation: setting $\tilde Q_1 = Q_1^2$ and $\tilde P_1 = P_1/(2Q_1)$ yields a fully rational canonical transformation that still separates the Hamilton--Jacobi equation.  For the elliptic coordinates no such reduction exists: both $Q_1$ and $Q_2$ depend on the same square root $S = \sqrt{T^2-4P}$, and any reparametrisation $\tilde Q_i = f_i(Q_i)$ that rationalises one coordinate leaves the other algebraic. Replacing $(Q_1,Q_2)$ by their symmetric functions $(T,P)$ is rational but destroys the Stäckel separation.  This is intrinsic to confocal elliptic coordinates, which are the two roots of a quadratic in the eigenvalue variable; the square root persists in any parametrisation, including the trigonometric form \cite[eq.~(66)]{PWY2017}.
\end{rem}

\paragraph{Level sets and connection to the sphere.}
Eliminating $S$ from $Q_1 = \tilde\lambda_1$ gives the level-set equation
\begin{equation}
\label{eq:confocal_level_sets_paper}
  \frac{\gamma_2\,q_1^2}{\tilde\lambda_1 - k_2^2}
  + \frac{\gamma_2 k_1^2\,q_2^2}{\tilde\lambda_1} = 1,
\end{equation}
a confocal ellipse for $\tilde\lambda_1 > k_2^2$ and a hyperbola for $0 < \tilde\lambda_1 < k_2^2$.
These are the gnomonic projections ($q_i = \xi_i/\xi_3$) of the spherical conics
\begin{equation}\label{eq:spherical_conics}
  \frac{\xi_1^2}{A} + \frac{k_1^2\,\xi_2^2}{B} = \xi_3^2, \qquad
  A = \frac{\tilde\lambda_1-k_2^2}{\gamma_2},\quad B = \frac{\tilde\lambda_1}{\gamma_2},
\end{equation}
on $S^2_R$; substituting $\xi_i = q_i\xi_3$ and dividing by $\xi_3^2$ recovers~\eqref{eq:confocal_level_sets_paper}.
The Haantjes machinery produces the confocal family without assuming in advance that elliptic coordinates are the correct ones.

\paragraph{Separated Hamiltonians and quantum ODE types.}
The DH coordinates produced by the Haantjes algorithm also separate the quantum Schrödinger equation for $H_{(2)}$ \cite{PWY2017, PSAWY2017}.  Under Weyl quantisation ($p_j\mapsto -i\hbar\partial_{q_j}$), the classical Hamiltonian \eqref{eq:hamZK} becomes the differential operator
\begin{equation}\label{eq:Zernike_Schrodinger}
  \hat{H}_{(2)} = -\hbar^2\Delta + \gamma_1(-i\hbar)\,\mathbf{q}\cdot\nabla + \gamma_2(-i\hbar)^2(\mathbf{q}\cdot\nabla)^2,
\end{equation}
where $\Delta = \partial_{q_1}^2+\partial_{q_2}^2$ and $\mathbf{q}\cdot\nabla = q_1\partial_{q_1}+q_2\partial_{q_2}$; the eigenvalue equation $\hat{H}_{(2)}\Psi = E\Psi$ separates in each of the coordinate systems below.  The type of the resulting ODE is determined by the structure of the separated Hamiltonian.

\begin{center}
\begin{tabular}{lll}
\hline
Separation & DH positions & Separated quantum ODE \\
\hline
Polar ($\mathcal{J}^2$) & $r$ (rational) & hypergeometric, 3 reg.\ sing.; Jacobi poly.\ \cite{PSAWY2017} \\
Cartesian ($I_2$) & $q_1,\; q_2/\sqrt{1+\gamma_2 q_1^2}$ (rational) & hypergeometric, 3 reg.\ sing.; Gegenbauer$\,{\times}\,$Legendre \cite{PSAWY2017} \\
Cartesian ($I_1$) & $q_1/\sqrt{1+\gamma_2 q_2^2},\; q_2$ (rational) & hypergeometric, 3 reg.\ sing.; Gegenbauer$\,{\times}\,$Legendre \cite{PSAWY2017} \\
Elliptic ($I_e$) & $\tfrac{1}{2}(-T\pm S)$ (algebraic, deg.~2) & Heun, 4 reg.\ sing.\ \cite{Ataki18,Ataki19b} \\
\hline
\end{tabular}
\end{center}

\begin{rem}[The Heun class]
A second-order linear ODE is said to belong to the \emph{Heun class} if it is Fuchsian with exactly four regular singular points (including $\infty$).  This is the first class beyond the hypergeometric equation (three regular singular points) and the generic member cannot be solved in terms of classical special functions.
\end{rem}

\noindent\emph{Polar.} The separated Hamiltonian~\eqref{eq:HN_polar} for $N=2$ is
\begin{equation}\label{eq:H2_polar_disc}
  H_{(2)} = (1 + \gamma_2 r^2)\,p_r^2 + \frac{p_\varphi^2}{r^2} + \gamma_1 r\,p_r.
\end{equation}
Upon quantisation, the radial eigenvalue equation has leading coefficient $(1+\gamma_2 r^2)$, which vanishes at the boundary $r = R = 1/\sqrt{|\gamma_2|}$, and a centrifugal term $+m^2/r^2$ singular at $r=0$. The three regular singular points $r = 0$, $r = R$, $r = \infty$ classify the ODE as hypergeometric; solutions are Jacobi polynomials \cite{PSAWY2017}.

\noindent\emph{Elliptic: separated Hamiltonian.}
In the DH coordinates $(\tilde\lambda_1,\tilde\lambda_2,P_1,P_2)$ with $\tilde\lambda_1>k_2^2>\tilde\lambda_2>0$,
a direct computation using the generating function $F_2 = P_1\tilde\lambda_1(q)+P_2\tilde\lambda_2(q)$
yields the Stäckel form
\begin{equation}
\label{eq:H2_separated_elliptic}
  H_{(2)} = \frac{h(\tilde\lambda_1)\,P_1^2 - h(\tilde\lambda_2)\,P_2^2
             + g(\tilde\lambda_1)\,P_1 - g(\tilde\lambda_2)\,P_2}{\tilde\lambda_1-\tilde\lambda_2},
\end{equation}
where
\begin{equation}
\label{eq:h_g_elliptic}
  h(\lambda) := 4\gamma_2\,\lambda(\lambda-k_2^2)(\lambda+k_1^2),
  \qquad
  g(\lambda) := 2\gamma_1\,\lambda(\lambda-k_2^2).
\end{equation}
The denominator $\tilde\lambda_1-\tilde\lambda_2 = S\neq 0$ on the open set where the coordinate map
is defined ($S=0$ is the focal locus).  The structure mirrors
\eqref{eq:H2_separated}--\eqref{eq:H2_separated_I1}: the Stäckel weight $h(\lambda)$ encodes the
confocal geometry, while $g(\lambda)$ carries the $\gamma_1$ correction.

\noindent\emph{Elliptic: quantum ODE type.}
Upon quantisation, the eigenvalue equation in the $\tilde\lambda_1$-variable has leading coefficient
$h(\tilde\lambda_1)/(\tilde\lambda_1-\tilde\lambda_2)$.  The polynomial $h(\lambda) = 4\gamma_2\lambda(\lambda-k_2^2)(\lambda+k_1^2)$
has three finite zeros at $\lambda=0$, $\lambda=k_2^2$, and $\lambda=-k_1^2$; together with
$\lambda\to\infty$ these give four regular singular points, placing the eigenvalue equation in the
Heun class \cite{Ataki18,Ataki19b}.

This is the most complex separation accessible within the $N=2$ EPT class; the next section shows that for $N\geq 3$ even the EPT structure itself breaks down.

\section{Extended point transformations for \texorpdfstring{$N \geq 3$}{N >= 3}}
\label{sec:no_ept}

\noindent The Cartesian-type DH coordinates constructed in Subsections~\ref{subsec4.2}--\ref{subsec4.3} are extended point transformations (EPTs): the new position coordinates $Q_i$ depend only on $(q_1, q_2)$, not on the momenta.  We now show that this feature is special to $N \leq 2$: for $N \geq 3$ the only EPT that brings $H_{(N)}$ to Stäckel form leads to polar coordinates, and the corresponding separated integral is $\mathcal{J}^2$ rather than $I_1$ or $I_2$.

\begin{Theo}
\label{thm:no_ept}
Let $N\geq 3$ and suppose $\gamma_k\neq 0$ for some $3\leq k\leq N$.  Any extended
point transformation $Q=Q(q)$ that brings $H_{(N)}$ to separated form
is a polar-type transformation: up to reparametrisation,
$(Q_1,Q_2)=(f(\varphi),g(r))$ for smooth functions $f$ and $g$.  The
corresponding separated first integral is $p_\varphi^2=\mathcal{J}^2$.
Consequently, the integrals $I_1,I_2$ cannot be separated by any extended
point transformation.
\end{Theo}

\begin{proof}
\textbf{Step 1: Stäckel condition.}
Under an EPT with Jacobian $J_{ij}=\partial Q_i/\partial q_j$, the canonical lift
$p=J^TP$ gives $q\cdot p = (Jq)\cdot P$. Setting $v=v(q)=Jq$, the Hamiltonian
in the new frame is
\begin{equation}\label{eq:H_EPT_frame}
  \tilde{H} = P^T JJ^T P + \sum_{n=1}^N \gamma_n(v\cdot P)^n.
\end{equation}
For $\tilde{H}$ to take separated form $f(Q_1,Q_2)\tilde{H}=\phi_1(Q_1,P_1)+\phi_2(Q_2,P_2)$
(the most general ansatz guaranteeing additive separability $W=W_1(Q_1)+W_2(Q_2)$ of the Hamilton--Jacobi equation, cf.\ the conformal Stäckel form in \cite{Benenti2016}), differentiating with respect to $P_1$
and $P_2$ gives
\be\label{eq:cross}
  \frac{\partial^2\tilde{H}}{\partial P_1\,\partial P_2}
  = 2(JJ^T)_{12} + v_1 v_2\sum_{n=2}^{N}n(n-1)\gamma_n(v\cdot P)^{n-2} = 0
  \quad\text{(as a polynomial in }P\text{).}
\ee

\textbf{Step 2: Forcing $v_1v_2=0$.}
Let $k=\max\{n\geq 3:\gamma_n\neq 0\}$.  The coefficient of $(v\cdot P)^{k-2}$
(degree $k-2\geq 1$) in \eqref{eq:cross} is $k(k-1)\gamma_k v_1 v_2$, so
\be\label{eq:vv}
  v_1(q)\,v_2(q) = 0 \quad\text{for all }q.
\ee
The degree-$0$ coefficient then gives $(JJ^T)_{12}=\nabla Q_1\cdot\nabla Q_2=0$.

\textbf{Step 3: Polar-type coordinates.}
Let $U_i=\{q:v_i(q)\neq 0\}$.  Since $J(q)$ is invertible and $q\neq 0$, we
have $Jq\neq 0$, so $v_1$ and $v_2$ cannot both vanish: $Z_1\cap Z_2=\emptyset$
on $\Omega\setminus\{0\}$.  By \eqref{eq:vv}, $U_1$ and $U_2$ are disjoint
open sets covering $\Omega\setminus\{0\}$.  Since $\Omega\subset\mathbb{R}^2$
is open and connected, so is $\Omega\setminus\{0\}$; hence one of
$U_1,U_2$ is empty.  Without loss of generality say $v_1\equiv 0$. Then
\begin{equation}\label{eq:Q1_radial_free}
  \nabla Q_1\cdot q =(q_1 \partial_{q_1} + q_2 \partial_{q_2})Q_1= r\,\partial_r Q_1 = 0 \implies \partial_r Q_1=0
  \implies Q_1=Q_1(\varphi).
\end{equation}
Since $\nabla Q_1\parallel\hat{e}_\varphi$, the condition $\nabla Q_1\cdot\nabla Q_2=0$
forces $\nabla Q_2\perp\hat{e}_\varphi$, i.e.\ $\partial_\varphi Q_2=0$, giving
$Q_2=Q_2(r)$.

\textbf{Step 4: The separated integral is $\mathcal{J}^2$.}
The EPT $(Q_1,Q_2)=(f(\varphi),g(r))$ is a reparametrisation of the polar
coordinates $(r,\varphi)$ constructed in Subsection~\ref{subsec4.1}.  As
established there, the Hamiltonian $H_{(N)}$ separates in polar coordinates
with $p_\varphi^2 = \mathcal{J}^2$ as the separated first integral
(equation~\eqref{eq:HN_polar}); any reparametrisation $Q_1=f(\varphi)$,
$Q_2=g(r)$ merely rescales the conjugate momenta
($P_1=p_\varphi/f'(\varphi)$, $P_2=p_r/g'(r)$) and leaves $p_\varphi^2$
unchanged as the invariant first integral.  Since $I_1$ and $I_2$ are
functionally independent of $H_{(N)}$ and $\mathcal{J}^2$, they cannot be
separated by any extended point transformation.
\end{proof}

\begin{rem}[Why $N=2$ is different]\label{rem:N2_different}
For $N=2$, condition~\eqref{eq:cross} reduces to the single equation
$2(JJ^T)_{12}+2\gamma_2 v_1v_2=0$, which can be satisfied with $J$
invertible (e.g.\ polar coordinates, where $v_1=\nabla\varphi\cdot q=0$ and
$(JJ^T)_{12}=0$).  The extra constraint~\eqref{eq:vv} is only forced when a
term $(v\cdot P)^{k-2}$ with $k\geq 3$ must vanish.
\end{rem}

\begin{rem}[Analytic EPTs]\label{rem:analytic_ept}
The topological argument in Step~3 can be simplified if the EPT is
assumed real-analytic (the standard setting in mechanics).  In that case
$v_1,v_2$ are real-analytic, and the conclusion $v_1\equiv 0$ or
$v_2\equiv 0$ follows directly from the identity theorem: if
$v_1\not\equiv 0$ then its zero set has empty interior, so $v_2=0$ on the
open dense set $\{v_1\neq 0\}$, hence $v_2\equiv 0$ by continuity.  This
argument requires neither invertibility of $J$ at $q=0$ nor the
separate connectivity step.
\end{rem}

\section{Conclusions and Perspectives}
\label{conc}
\noindent We have applied the symplectic-Haantjes framework to construct, by an algorithmic two-step procedure, explicit Darboux--Haantjes (DH) separation coordinates for the classical Zernike system $H_{(2)}$, as a first step in the program of applying this method to the full generalized family $H_{(N)}$ \eqref{eq:HN}. The three independent constants of motion --- the angular momentum $\mathcal{J}$ and the two quadratic integrals $I_1$, $I_2$ --- each determine a Haantjes operator and a corresponding canonical transformation to separation coordinates.

\medskip

\noindent\textbf{Polar separation.} The Haantjes operator associated with $\mathcal{J}^2$ (equation~\eqref{eq:haantjes_pphi}) is independent of all parameters $\gamma_n$. The DH algorithm produces the standard polar coordinates $(r, \varphi, p_r, p_\varphi)$, which separate the Hamilton--Jacobi equation for the entire generalized family $\{H_{(N)}\}_{N \geq 1}$ simultaneously. As noted in Remark~\ref{rem:curvature_blind}, the Haantjes construction is insensitive to the curvature $\kappa = -\gamma_2$ of the underlying space: both the flat radial coordinate $r$ and the geodesic distance $\rho$ are admissible characteristic coordinates, and the geometric distinction between them must be supplied by additional input beyond the Haantjes structure.

\medskip

\noindent\textbf{Cartesian-type separations.} The Haantjes operators associated with $I_2$ and $I_1$ (equations~\eqref{eq:haantjes_I2} and \eqref{eq:haantjes_I1_cart}) each admit a Nijenhuis generator whose eigenvalues depend only on the coordinates $q$. This momentum-independence implies that the corresponding DH coordinates are reachable by extended point transformations (EPTs). The explicit transformations \eqref{eq:canonical_N2} and \eqref{eq:canonical_I1} are related by the $q_1 \leftrightarrow q_2$ symmetry of $H_{(2)}$, and both diagonalise the Hamilton--Jacobi equation via separated integrals of the form
\begin{equation}\label{eq:Ik_separated}
  I_k = P_k^2+\sum_{n=1}^{2} \gamma_n (Q_k P_k)^n, \qquad k=1,2.
\end{equation}
As shown in Remark~\ref{rem:geometric_interpretation}, the two transformations correspond geometrically to the System~II (spherical) and System~HII (equidistant) coordinates of Pogosyan, Wolf, and Yakhno~\cite{PWY2017}, here recovered by purely algebraic means.

\medskip

\noindent\textbf{Elliptic separation.} The separation constant $I_e = -\gamma_2 k_1^2\,\mathcal{J}^2 - k_2^2\,I_2$ (where $k_1 = \cos f$, $k_2 = \sin f$ parametrise the interfocal distance) is a linear combination of known integrals \cite[eq.~(72)]{PWY2017}. By linearity of the chain equation~\eqref{eq:chain_equations}, the operator $K_e = -\gamma_2 k_1^2\,K_{\mathcal{J}^2} - k_2^2\,K_{I_2}$ automatically satisfies $K_e^T\,dH_{(2)} = dI_e$.  However, the individual $\omega\mathscr{H}$ structures associated with $\mathcal{J}^2$ and $I_2$ are distinct, and $K_e$ does not belong to either of them: that $K_e$ is itself a Haantjes operator is a non-trivial fact, since the Haantjes condition is nonlinear in $K$.  Proposition~\ref{prop:Ke_haantjes} verifies this by direct computation and establishes the corresponding semisimple Abelian $\omega\mathscr{H}$ structure for the pair $(H_{(2)}, I_e)$.  Since $K_e$ is a linear combination of lift-form operators it is again of lift form ($B=0$), so Proposition~\ref{prop:a_block} applies and the DH coordinates can be reached by an EPT.  We use the Nijenhuis eigenvalues as coordinates (see Remark~\ref{rem:nijenhuis_simplification}); these are non-negative and give the same separation as the $A$-block eigenvalues (which differ by an overall sign, absorbed into the coordinate choice):
\begin{equation}\label{eq:Q12_elliptic}
  Q_{1,2} = \tilde\lambda_{1,2} = \tfrac{1}{2}\bigl(T \pm \sqrt{T^2-4P}\bigr), \qquad \tilde\lambda_1 > \tilde\lambda_2 \geq 0.
\end{equation}
Unlike the Cartesian case, these coordinates are algebraic of degree~$2$ over $\mathbb{Q}(q_1,q_2)$ and cannot be expressed as rational functions of the original coordinates. The level sets of $\tilde\lambda_1$ are confocal conics in the $(q_1,q_2)$-plane, arising as gnomonic projections of spherical conics on $S^2_R$ without any a priori assumption on the coordinate form.

The algebraic degree of the DH coordinate change is reflected in the type of the quantum separated ODE. For the elliptic case, the degree-$2$ coordinate map has a branch point that manifests as a fourth regular singular point, placing the ODE in the Heun class \cite{Ataki18, Ataki19b}. For the polar case, the centrifugal term $p_\varphi^2/r^2$ in the separated Hamiltonian~\eqref{eq:HN_polar} produces a regular singular point at $r = 0$; together with the boundary singularity at $r = R$ the ODE is hypergeometric with three regular singular points \cite{PSAWY2017}. In both cases the ODE type is read off from the explicit separated Hamiltonian produced by the Haantjes algorithm.

\medskip

\noindent\textbf{Obstruction for $N \geq 3$.} Theorem~\ref{thm:no_ept} establishes that for $N \geq 3$ (with $\gamma_N \neq 0$), no extended point transformation can reach the DH separation coordinates of $K_{I_1}$ or $K_{I_2}$. This places the $N \leq 2$ Zernike Hamiltonian in a distinguished position within the generalized family: the EPT structure found in Subsections~\ref{subsec4.2}--\ref{subsec:elliptic} is a special feature of the quadratic case and cannot be extended to higher-order Zernike Hamiltonians. We stress, however, that Theorem~\ref{thm:jacobi_haantjes} continues to guarantee the existence of DH separation coordinates for any $N$; the obstruction concerns only the \emph{class} of the canonical transformation, not the separability itself.

\medskip

\noindent\textbf{Open problems.} Several natural directions remain open.

\begin{itemize}
\item \emph{Separation for higher-order Zernike systems.}
While Theorem~\ref{thm:no_ept} precludes extended-point DH constructions for $N \ge 3$, the possibility of momentum-dependent canonical transformations remains open. For the generalized Zernike family $H_{(N)}$ with $N \geq 3$, which is superintegrable for arbitrary $N$ \cite{AGSH, Gonera22}, a separation has yet to be established. The key structural difficulty is that the relevant integrals of motion are polynomial of degree $N$ in the momenta, so the standard Stäckel--Benenti theory \cite{BenentiFrancaviglia1980, Benenti2016} is inapplicable.

\item \emph{Higher-dimensional Zernike systems.} The Zernike system admits a natural $n$-dimensional generalization $H = |\boldsymbol{p}|^2 + F(\boldsymbol{q}\cdot\boldsymbol{p})$ on $\mathbb{R}^n$, which is maximally superintegrable \cite{Gonera22}. The construction of Haantjes algebras and separation coordinates for these higher-dimensional systems is an open problem.

\item \emph{Haantjes geometry and extended point transformations.} The obstruction result (Theorem~\ref{thm:no_ept}) is formulated for a specific two-dimensional family. We expect analogous obstructions to arise for higher-order integrals of natural Hamiltonians $T + V(q)$ or magnetic Hamiltonians $|\boldsymbol{p} - \boldsymbol{A}(q)|^2 + V(q)$, in any dimension. A general formulation of this obstruction, and a fuller understanding of the relationship between Haantjes structures, Killing tensors, and extended point transformations, deserves further investigation.
\end{itemize}

\section*{Declaration of generative AI and AI-assisted technologies in the writing process}
During the preparation of this work, the authors used ChatGPT by OpenAI and Claude by Anthropic in order to identify grammatical errors and to aid in the overall readability of the text. After using this tool, the output was reviewed and edited as needed and the authors take full responsibility for the content of the publication.


\section*{Acknowledgements}

\phantomsection
\addcontentsline{toc}{section}{Acknowledgements}

\noindent We thank Piergiulio Tempesta for a careful reading and useful comments on an earlier version of this manuscript.

\medskip

\noindent O.K.’s postdoctoral fellowship is financed by the project “Fostering ICMAT’s Strategic Scientific Lines” (202450E223). 
D.L. is supported by HORIZON EUROPE - European Research Council (ERC) - STARTING GRANT 2021 “Hamiltonian Dynamics, Normal Forms and Water Waves” (HamDyWWa) - Project Number: HE$\_$ERC22RMONT$\_$01.
Views and opinions expressed are however those of the authors only and do not necessarily reflect those
of the European Union or the European Research Council. Neither the European Union nor the granting authority can be held responsible for them.
The research of D.L.~has also been partially funded by MUR - Dipartimento di Eccellenza 2023-2027, codice CUP G43C22004580005 - codice progetto   DECC23$\_$012$\_$DIP and partially supported by INFN-CSN4 (Commissione Scientifica Nazionale 4 - Fisica Teorica), MMNLP project. D.L. is a member of GNFM, INdAM.



\begin{thebibliography}{99}
\small

\phantomsection
\addcontentsline{toc}{section}{References}

\bibitem{PWY2017}
G. S. Pogosyan, K. B. Wolf, and A. Yakhno, Superintegrable classical Zernike system, \href{https://doi.org/10.1063/1.4990793}{J. Math. Phys. \textbf{58} (2017) 072901}.

\bibitem{Zernike}
F. von Zernike,   Beugungstheorie des Schneidenverfahrens und seiner verbesserten Form, der Phasenkontrastmethode, \href{https://doi.org/10.1016/S0031-8914(34)80259-5}{Physica \textbf{1} (1934) 689}.

 \bibitem{PSAWY2017}
G. S. Pogosyan, C. Salto-Alegre, K. B. Wolf, A. Yakhno, Quantum superintegrable Zernike system, \href{https://doi.org/10.1063/1.4990794}{J. Math. Phys. \textbf{58} (2017) 072101}.

\bibitem{Das01}
C. Daskaloyannis, Quadratic Poisson algebras of two-dimensional classical superintegrable systems and quadratic associative algebras of quantum superintegrable systems, \href{https://doi.org/10.1063/1.1348026}{J. Math. Phys. \textbf{42} (2001), 1100}.

\bibitem{AGSH}
A. Blasco, I. Gutierrez-Sagredo, F. J. Herranz, Higher-order superintegrable momentum-dependent Hamiltonians on curved spaces from the classical Zernike system, \href{https://doi.org/10.1088/1361-6544/acad5e}{Nonlinearity \textbf{36} (2023) 1143}.

\bibitem{ANY26}
V. Abgaryan, A. Nersessian, V. Yeghikyan, Zernike system revisited: imaginary gauge and Higgs oscillator, \href{https://doi.org/10.1142/S0217732325502220}{Mod. Phys. Lett. A \textbf{41} (2026) 2550222}

  \bibitem{Fordy2018}
A. P. Fordy, Classical and quantum super-integrability: From Lissajous figures to exact solvability, \href{https://doi.org/10.1134/S1063778818060133}{Phys. Atom. Nuclei \textbf{81} (2018) 832}. 

\bibitem{Higgs79}
P. W. Higgs, Dynamical symmetries in a spherical geometry I. \href{https://doi.org/10.1088/0305-4470/12/3/006}{J. Phys. A: Math. Gen. \textbf{12} (1979) 309–323}

\bibitem{bulg}
F. J. Herranz, A. Blasco, R. Campoamor-Stursberg, I. Gutierrez-Sagredo, D. Latini, I. Marquette, Generalized classical and quantum Zernike Hamiltonians, \href{https://doi.org/10.55318/bgjp.2025.52.s1.139}{Bulg. J. Phys. \textbf{52-s1} (2025) 139-145}

\bibitem{nozaleda2022classical}
D.~Reyes~Nozaleda, P.~Tempesta, G.~Tondo, Classical multiseparable Hamiltonian systems, superintegrability and Haantjes geometry, \href{https://doi.org/10.1016/j.cnsns.2021.106021}{Commun. Nonlinear Sci. Numer. Simulat. \textbf{104} (2022) 106021}.

\bibitem{Campoamor13}
R.~Campoamor-Stursberg, J.~F.~Cariñena, M.~F.~Rañada, Higher-order superintegrability of a Holt-related potential, \href{https://doi.org/10.1088/1751-8113/46/43/435202}{J. Phys. A \textbf{46} (2013) 435202}.

\bibitem{tempesta2022haantjes}
P.~Tempesta, G.~Tondo, Haantjes algebras of classical integrable systems, \href{https://doi.org/10.1007/s10231-021-01107-4}{Ann. Mat. Pura Appl. \textbf{201} (2022) 57--90}.

\bibitem{tempesta2021haantjes}
P.~Tempesta, G.~Tondo, Haantjes algebras and diagonalization, \href{https://doi.org/10.1016/j.geomphys.2020.103968}{J. Geom. Phys. \textbf{160} (2021) 103968}.

\bibitem{haantjes1955}
J.~Haantjes, On $X_m$-forming sets of eigenvectors, \href{https://doi.org/10.1016/S1385-7258(55)50021-7}{Indag. Math. \textbf{17} (1955) 158--162}.

\bibitem{magri1984}
F.~Magri, C.~Morosi, A geometrical characterization of integrable Hamiltonian systems through the theory of Poisson--Nijenhuis manifolds, Quaderno S~\textbf{19}, Università di Milano (1984).

\bibitem{magri2003}
F.~Magri, Lenard chains for classical integrable systems, \href{https://doi.org/10.1023/B:TAMP.0000007919.80743.1e}{Theoret. Math. Phys. \textbf{137} (2003) 1716--1722}.

\bibitem{morosi1997}
C.~Morosi, G.~Tondo, Quasi-bi-Hamiltonian systems and separability, \href{https://doi.org/10.1088/0305-4470/30/8/023}{J. Phys. A \textbf{30} (1997) 2799--2806}.

\bibitem{morosi1998}
C.~Morosi, G.~Tondo, On a class of dynamical systems both quasi-bi-Hamiltonian and bi-Hamiltonian, \href{https://doi.org/10.1016/S0375-9601(98)00543-X}{Phys. Lett. A \textbf{247} (1998) 59--64}.

\bibitem{tondo1999}
G.~Tondo, C.~Morosi, Bi-Hamiltonian manifolds, quasi-bi-Hamiltonian systems and separation of variables, \href{https://doi.org/10.1016/S0034-4877(99)80167-0}{Rep. Math. Phys. \textbf{44} (1999) 255--266}.

\bibitem{Reyes2024}
D.~Reyes, P.~Tempesta, G.~Tondo, Partial separability and symplectic-Haantjes manifolds, \href{https://doi.org/10.1007/s10231-024-01462-y}{Ann. Mat. Pura Appl. (4) \textbf{203} (2024) 2677--2710}.

\bibitem{Benenti2016}
S.~Benenti, Separability in Riemannian Manifolds, \href{https://doi.org/10.3842/SIGMA.2016.013}{SIGMA \textbf{12} (2016) 013, 21~pp.}

\bibitem{BenentiFrancaviglia1980}
S.~Benenti, M.~Francaviglia, The theory of separability of the Hamilton--Jacobi equation and its applications to general relativity, in \textit{General Relativity and Gravitation}, Vol.~1, ed. A.~Held, Plenum Press, New York, 1980, pp.~393--439.

\bibitem{MPW}
W.~Miller Jr., S.~Post, and P.~Winternitz, Classical and quantum superintegrability with applications, \href{https://doi.org/10.1088/1751-8113/46/42/423001}{J. Phys. A: Math. Theor. \textbf{46} (2013) 423001}.

\bibitem{Ataki19}
N.~M. Atakishiyev, G.~S. Pogosyan, L.~E. Vicent, K.~B. Wolf, and A.~Yakhno,
Spherical geometry, Zernike's separability, and interbasis expansion coefficients,
\href{https://doi.org/10.1063/1.5099974}{J. Math. Phys. \textbf{60} (2019) 101701}.

\bibitem{Ataki18}
N.~M. Atakishiyev, G.~S. Pogosyan, L.~E. Vicent, K.~B. Wolf, and A.~Yakhno,
Elliptic basis for the Zernike system: Heun function solutions,
\href{https://doi.org/10.1063/1.5030759}{J. Math. Phys. \textbf{59} (2018) 073503}.

\bibitem{Gonera22}
C.~Gonera, J.~Gonera, and P.~Kosiński,
On the generalization of classical Zernike system,
\href{https://doi.org/10.1088/1361-6544/ad1b8d}{Nonlinearity \textbf{37} (2024) 025019}.

\bibitem{Campoamor25}
R.~Campoamor-Stursberg, F.~J.~Herranz, D.~Latini, I.~Marquette, and A.~Blasco,
Generalized quantum Zernike Hamiltonians: polynomial Higgs-type algebras and algebraic derivation of the spectrum,
\href{https://doi.org/10.1088/1361-6544/ae64a2}{Nonlinearity \textbf{39} (2026) 055008}.

\bibitem{Pog17}
G.~S. Pogosyan, K.~B. Wolf, and A.~Yakhno,
New separated polynomial solutions to the Zernike system on the unit disk and interbasis expansion,
\href{https://doi.org/10.1364/JOSAA.34.001844}{J. Opt. Soc. Am. A \textbf{34} (2017) 1844}.

\bibitem{Ataki17}
N.~M. Atakishiyev, G.~S. Pogosyan, L.~E. Vicent, K.~B. Wolf, and A.~Yakhno,
Interbasis expansions in the Zernike system,
\href{https://doi.org/10.1063/1.5000915}{J. Math. Phys. \textbf{58} (2017) 103505}.

\bibitem{Ataki19b}
N.~M. Atakishiyev, G.~S. Pogosyan, L.~E. Vicent, K.~B. Wolf, and A.~Yakhno,
On elliptic trigonometric form of the Zernike system and polar limits,
\href{https://doi.org/10.1088/1402-4896/aafecb}{Phys. Scr. \textbf{94} (2019) 045202}.

\bibitem{KKMbook}
E.~G. Kalnins, J.~M. Kress, and W.~Miller Jr.,
\textit{Separation of Variables and Superintegrability: The Symmetry of Solvable Systems},
IOP Publishing, Bristol, 2018, ISBN 978-0-7503-1314-8.

\bibitem{TTW09}
F.~Tremblay, A.~V. Turbiner, and P.~Winternitz,
An infinite family of solvable and integrable quantum systems on a plane,
\href{https://doi.org/10.1088/1751-8113/42/24/242001}{J. Phys. A: Math. Theor. \textbf{42} (2009) 242001}.

\bibitem{PW10}
S.~Post and P.~Winternitz,
An infinite family of superintegrable deformations of the Coulomb potential,
\href{https://doi.org/10.1088/1751-8113/43/22/222001}{J. Phys. A: Math. Theor. \textbf{43} (2010) 222001}.

\end{thebibliography}
\end{document}